\documentclass[onecolumn,journal,twoside]{IEEEtran}

\usepackage{microtype}
\usepackage{graphicx}
\usepackage{subfigure}
\usepackage{booktabs} % for professional tables

\usepackage[table,xcdraw]{xcolor}

\usepackage{pgfplots}
\usepackage{bm}
\usepackage{enumerate}
\usepackage{enumitem}
\newcommand{\vct}[1]{\bm{#1}}
\usepackage{algorithm}
\usepackage{algpseudocode}

\usepackage{amsthm}
\newcounter{motivation}
\newtheorem{theorem}{Theorem}

\newtheorem{lemma}{Lemma}

\theoremstyle{remark}
\newtheorem{remark}{Remark}

\usepackage{cite}
\usepackage{bbm}
\usepackage{amsfonts}
\usepackage{amsmath}
\usepackage{amssymb}
\usepackage{hhline}
\usepackage[utf8]{inputenc}
\usepackage[english]{babel}
\usepackage{multirow}
 \newcommand{\bF}{\mathbb{F}}
 \newcommand{\bV}{\mathbb{V}}
 \newcommand{\bU}{\mathbb{U}}
 
 \newcommand{\bR}{\mathbb{R}}
 
 \newcommand{\cS}{\mathcal{S}}
 \newcommand{\cT}{\mathcal{T}}
 
\DeclareMathOperator{\avg}{avg}
\DeclareMathOperator{\diag}{diag}
\usepackage{breakcites}
\usepackage{hyperref}
\newcommand{\Removed}[1]{}
\newcommand{\Added}[1]{{#1}}

\begin{document}

% If your paper is accepted and the title of your paper is very long,
% the style will print as headings an error message. Use the following
% command to supply a shorter title of your paper so that it can be
% used as headings.
%
%\runningtitle{Lagrange Coded Computing}

% If your paper is accepted and the number of authors is large, the
% style will print as headings an error message. Use the following
% command to supply a shorter version of the authors names so that
% they can be used as headings (for example, use only the surnames)
%
%\runningauthor{Surname 1, Surname 2, Surname 3, ...., Surname n}

\sloppy
% \setlength{\abovedisplayskip}{1mm}
% \setlength{\belowdisplayskip}{1mm}
% \setlength{\abovecaptionskip}{1mm}
% \setlength{\belowcaptionskip}{1mm}

%\twocolumn[

%\aistatstitle{Lagrange Coded Computing: Optimal Design for Resiliency, Security, and Privacy}

%\aistatsauthor{ Anonymous Authors }

%\aistatsaddress{ Institution 1 \And  Institution 2 \And Institution 3 } 
%\vspace{5mm}
%]

%TODO!!!1. deocoding 2. algorithm

\title{Lagrange Coded Computing: Optimal Design for Resiliency, Security, and Privacy}

	\author{Qian~Yu$^{*}$, Songze Li$^{*}$, Netanel Raviv$^{\dagger}$, Seyed Mohammadreza Mousavi Kalan$^{*}$, Mahdi Soltanolkotabi$^{*}$, and A.~Salman~Avestimehr$^{*}$\\
		$^{*}$ Department of Electrical Engineering, University of Southern California, Los Angeles, CA, USA \\ 
		$^{\dagger}$ Department of Electrical Engineering, California Institute of Technology, Pasadena,  CA, USA\\
	}

\maketitle
%\vskip 0.3in

\begin{abstract}%%% SHOULD BE 4-6 SENTENCES LONG %%%
We consider a  scenario involving computations over a massive dataset stored distributedly across multiple workers, which is at the core of distributed learning algorithms. 
We propose \textit{Lagrange Coded Computing} (LCC), a new framework to \emph{simultaneously} provide (1) \textit{resiliency} against stragglers that may prolong computations; (2) \textit{security} against Byzantine (or malicious) workers that deliberately modify the computation for their benefit; and (3) (information-theoretic) \textit{privacy} of the dataset amidst possible collusion of workers. 
LCC, which leverages the well-known Lagrange polynomial to create computation redundancy in a novel coded form across workers, can be applied to any computation scenario in which the function of interest is an \textit{arbitrary multivariate polynomial} of the input dataset, hence covering many computations of interest in machine learning. 
LCC significantly generalizes prior works to go beyond linear computations. It also enables secure and private computing in distributed settings, improving the computation and communication efficiency of the state-of-the-art. 
Furthermore, we prove the optimality of LCC by showing that it achieves the optimal tradeoff between resiliency, security, and privacy, i.e., in terms of tolerating the maximum number of stragglers and adversaries, and providing data privacy against the maximum number of colluding workers.
Finally, we show via experiments on Amazon EC2 that LCC speeds up the conventional uncoded implementation of distributed least-squares linear regression by up to $13.43\times$, and also achieves a $2.36\times$-$12.65\times$ speedup over the state-of-the-art straggler mitigation strategies.
\end{abstract}

% % this must go after the closing bracket ] following \twocolumn[ ...

% % This command actually creates the footnote in the first column
% % listing the affiliations and the copyright notice.
% % The command takes one argument, which is text to display at the start of the footnote.
% % The \sysmlEqualContribution command is standard text for equal contribution.
% % Remove it (just {}) if you do not need this facility.

% %\printAffiliationsAndNotice{}  % leave blank if no need to mention equal contribution
% \printAffiliationsAndNotice{\sysmlEqualContribution} % otherwise use the standard text

\section{Introduction}\label{section:introduction}
The massive size of modern datasets necessitates computational tasks to be performed in a distributed fashion, where the data is dispersed among many servers that operate in parallel~\cite{abadi2016tensorflow}. As we ``scale out'' computations across many servers, however, several fundamental challenges arise. Cheap commodity hardware tends to vary greatly in computation time, and it has been demonstrated~\cite{dean2013tail,LiMu,MultiTask} that a small fraction of servers, referred to as \textit{stragglers}, can be~$5$ to~$8$ times slower than the average, thus creating significant delays in computations.  Also, as we distribute computations across many servers, massive amounts data must be moved between them to execute the computational tasks, often over many iterations of a running algorithm, and this creates a substantial bandwidth bottleneck~\cite{li2014communication}. Distributed computing systems are also much more susceptible to adversarial servers, making security and privacy a major concern~\cite{oneAdvarsarialStraggler,CramerMPC,ShareMind}. 

We consider a general scenario in which the computation is carried out distributively across several workers, and propose \textit{Lagrange Coded Computing} (LCC), a new framework to simultaneously provide 

\begin{enumerate}[leftmargin=*] %[itemsep=-0.5ex]
\item \textit{resiliency} against straggler workers that may prolong computations;

\item \textit{security} against Byzantine (or malicious, {adversarial}) workers, with no computational restriction, that deliberately send erroneous data in order to affect the computation for their benefit; and

\item \textit{(information-theoretic) privacy} of the dataset amidst possible collusion of workers.
\end{enumerate}

\begin{figure}
  \centering
    \includegraphics[width=0.48\textwidth]{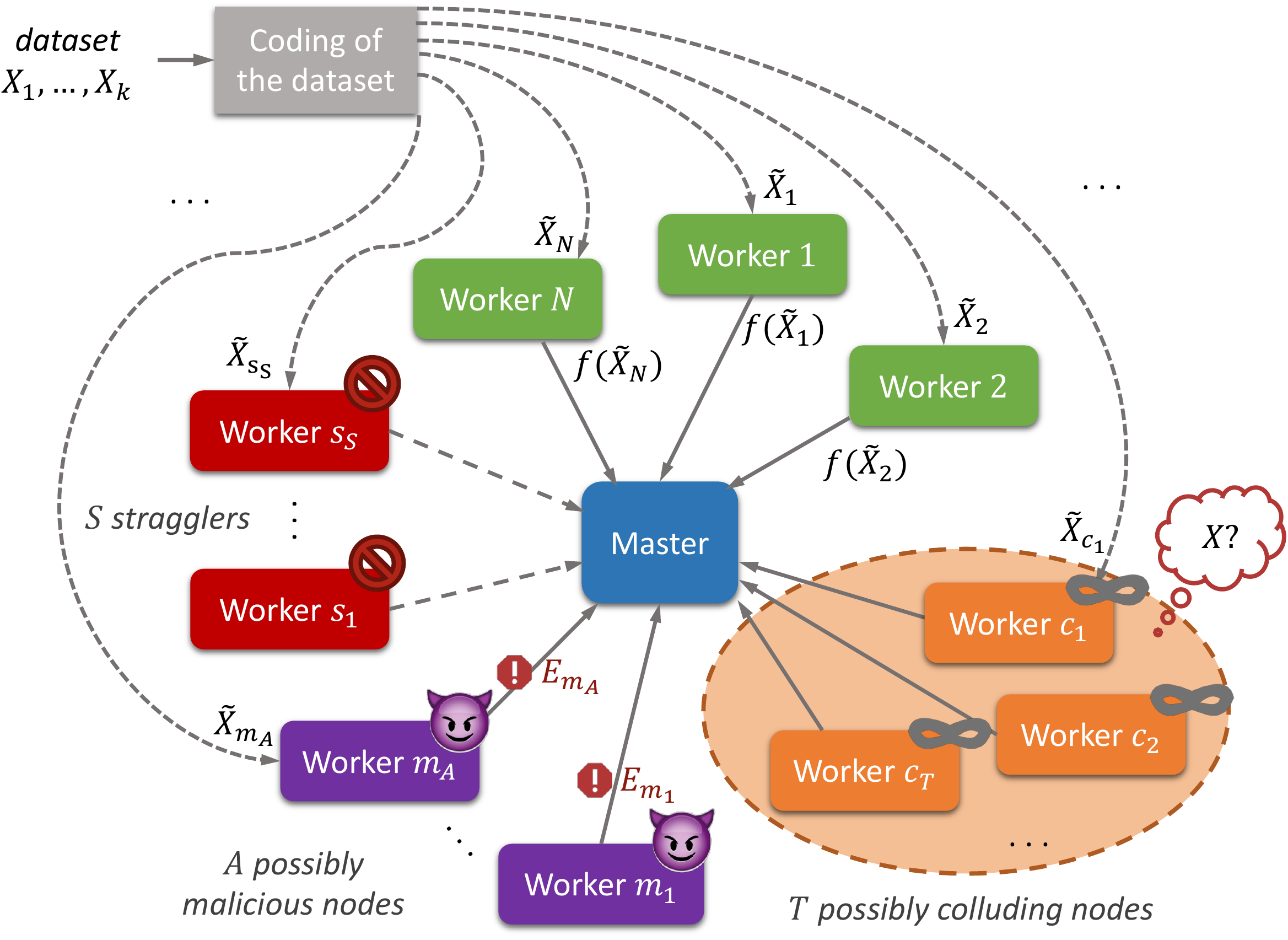}
 
    \caption{{\small An overview of the problem considered in this paper, where the goal is to evaluate a \textit{not necessarily linear} function~$f$ on a given dataset $X=(X_1,X_2,\ldots,X_K)$ using $N$ workers. Each worker applies~$f$ on a possibly \textit{coded} version of the inputs (denoted by $\tilde{X}_i$'s). By carefully designing the coding strategy, the master can decode all the required results from a subset of workers, in the presence of \textit{stragglers} (workers~$s_1,...,s_S$) and \textit{Byzantine workers} (workers~$m_1,...,m_A$), while keeping the dataset private to \textit{colluding workers} (workers~$c_1,...,c_T$).}}
  \label{fig:sys}

\end{figure}

%Lagrange coded computing 
LCC can be applied to any computation scenario in which the function of interest is an \textit{arbitrary multivariate polynomial} of the input dataset. This  covers many computations of interest in machine learning, such as various gradient and loss-function computations in learning algorithms and tensor algebraic operations (e.g., low-rank tensor approximation). The key idea of LCC is to encode the input dataset using the well-known Lagrange polynomial, in order to create computational redundancy in a novel coded form across the workers. This redundancy can then be exploited to provide resiliency to stragglers, security against malicious servers, and privacy of the dataset.

%\blue
% {Specifically, as illustrated in Fig.~\ref{fig:sys}, 
% we consider a setting where the goal is to compute~$f(X_i)$ for every~$X_i$ in a large dataset $X=(X_1,X_2,\ldots,X_K)$ using {a given number of workers}~$N$, where~$f$ is a given multilinear polynomial with degree $\deg f$. For a given~$N$ and~$f$, we say that the tuple~$(S,A,T)$ is \textit{achievable} if there exists a scheme that can complete the computations  in the presence of up-to~$S$ stragglers, up-to~$A$ adversarial workers, whilst keeping the data set private against sets of up-to~$T$ colluding workers\footnote{While our resiliency and security guarantees hold over every large enough field, perfect privacy is well defined only over finite fields. Hence, over infinite fields, our results hold for~$T=0$.}. By using Lagrange Coded Computing, we prove that~$(S,A,T)$ is achievable if~$N\ge (K+T-1)\deg f+S+2A+1$}.

{Specifically, as illustrated in Fig.~\ref{fig:sys}, using a master-worker distributed computing architecture with $N$ workers, the goal is to compute~$f(X_i)$ for every~$X_i$ in a large dataset $X=(X_1,X_2,\ldots,X_K)$, where~$f$ is a given multivariate polynomial with degree $\deg f$. 
To do so,~$N$ coded versions of the input dataset, denoted by $\tilde{X}_1,\tilde{X}_2,\ldots,\tilde{X}_N$ are created, and the workers then compute~$f$ over the coded data, {\textit{as if no coding is taking place}}. 
For a given~$N$ and~$f$, we say that the tuple~$(S,A,T)$ is \textit{achievable} if there exists an encoding and decoding scheme that can complete the computations  in the presence of up to~$S$ stragglers, up to~$A$ adversarial workers, whilst keeping the dataset private against sets of up to~$T$ colluding workers. %\footnote{While our resiliency and security guarantees hold over every large enough field, perfect privacy is well defined only over finite fields. Hence, over infinite fields, our results hold for~$T=0$.}. 

Our main result is that by carefully encoding the dataset the proposed LCC achieves~$(S,A,T)$ if~$(K+T-1)\deg f+S+2A+1 \leq N$. The significance of this result is that by  one additional worker (i.e., increasing $N$ by 1) LCC can increase the resiliency to stragglers by 1 or increase the robustness to malicious servers by $1/2$, while maintaining the privacy constraint. Hence, this result essentially extends the well-known optimal scaling of error-correcting codes (i.e., adding one parity can provide robustness against one erasure or $1/2$ error in optimal maximum distance separable codes) to the distributed secure computing paradigm.

%Using LCC, we prove that~$(S,A,T)$ is achievable if~$N\ge (K+T-1)\deg f+S+2A+1$}.

%We prove that in a setting where there are up-to $S$ stragglers, up-to $A$ adversarial workers, and the data set should be kept private against up-to~$T$ colluding workers, Lagrange Coded Computing can complete the computations by waiting for the results of any $(K+T-1)\deg f+S+2A+1$ workers (hence, tolerating the remaining stragglers). 
%Note that the number of workers the master needs to wait for does \emph{not} scale with the total number of workers $N$, hence the key property of LCC is that adding one additional worker can increase its resiliency to stragglers by 1 or increase its robustness to malicious servers by $1/2$, while maintaining the privacy constraint. Hence, this result essentially extends the well-known optimal scaling of error-correcting codes (i.e., adding one parity can provide robustness against one erasure or $1/2$ error in optimal maximum distance separable codes) to the distributed computing paradigm.
%To the best of our knowledge, there is no other scheme that achieves this optimal scaling. 

We prove the optimality of LCC {by showing that it achieves the optimal tradeoff between resiliency, security, and privacy. In other words, any computing scheme (under certain complexity constrains on the encoding and decoding designs) can achieve $(S,A,T)$ if and only if $(K+T-1)\deg f+S+2A+1 \leq N$.\footnote{{More accurately, when $N<K\textup{deg}f-1 $, we prove that the optimal tradeoff is instead given by $ K(S+2A+\textup{deg}\ f  \cdot T + 1)\leq N$, which can be achieved by a variation of the LCC scheme, as described in Appendix \ref{app:ulcc}. }}
%, among  computing schemes with linear coding designs.
This result further extends the scaling law in coding theory to private computing, showing that any additional worker enables data privacy against $1/\textup{deg} f$ additional colluding workers. }

%similar to the resiliency and security requirement, it is fundamental that any additional worker increases data privacy against colluding workers by $1/\deg f$. }}

%in regard to several metrics of interest. First, we prove %\Qian
%{a matching converse on the resiliency requirement, showing that LCC (and its uncoded version) can tolerate the maximum possible number of stragglers, among all schemes where both encoding and decoding are linear. Then by extending this converse, we show that LCC also provides security against the maximum possible number of adversaries, when $f$ is a multilinear function. }
%We also show that LCC requires injecting the minimum amount of randomness, among all computing schemes that universally achieve the same resiliency-security-privacy tradeoff for all linear functions $f$.
}

Finally, we specialize our general theoretical guarantees for LCC in the context of least-squares linear regression, which is one of the elemental learning tasks, and demonstrate its performance gain by optimally suppressing stragglers. Leveraging the algebraic structure of gradient computations, several strategies have been developed recently to exploit data and gradient coding for straggler mitigation in the training process  (see, e.g.,~\cite{speeding,TandonLDK17,maityrobust,karakus2017straggler,li2017near}). We implement LCC for regression on Amazon EC2 clusters, and empirically compare its performance with the conventional uncoded approaches, and two state-of-the-art straggler mitigation schemes: gradient coding (GC)~\cite{TandonLDK17,Halbawi,raviv2017gradient,ye2018communication} and matrix-vector multiplication (MVM) based approaches~\cite{speeding,maityrobust}. 
%In particular, we focus on the task of solving least-squares regression problems via gradient descent, which is one of the most fundamental learning tasks. The elegant algebraic structure of gradient computation for regression has drawn extensive studies on coded computing strategies (see, e.g.,~\cite{speeding,maityrobust,karakus2017straggler,TandonLDK17,li2017near}). 
Our experiment results demonstrate that compared with the uncoded scheme, LCC improves the run-time by $6.79\times$-$13.43\times$. Compared with the GC scheme, LCC improves the run-time by $2.36\times$-$4.29\times$. Compared with the MVM scheme, LCC improves the run-time by $1.01\times$-$12.65\times$.

% Compared with the GC scheme, LCC improves the run-time by up to $2.36\times$ with naturally occurring stragglers, and by up to $4.29\times$ with artificially introduced stragglers. Compared with the MVM scheme, LCC improves the run-time by up to $1.01\times$ with naturally occurring stragglers, and by up to $12.65\times$ with artificially introduced stragglers. 

%Compared with conventional uncoded approaches, LCC improves the runtime by %\textbf{Songze/Qian: please complete}

%\noindent~
%\paragraph{Related works.}
\textbf{Related works.}
%Related works.
There has recently been a surge of interest on using coding theoretic approaches to alleviate key bottlenecks (e.g., stragglers, bandwidth, and security) in distributed machine learning applications (e.g., \cite{lee2015speeding,LMA_all,yu2017optimally,Li2018fundamental,dutta2016short, NIPS2017_7027, TandonLDK17,Halbawi,raviv2017gradient,dutta2018on,MaddahAli,chen2018draco}). As we discuss in more detail in Section~\ref{sec:prior}, the proposed LCC scheme significantly advances prior works in this area by 1) generalizing coded computing to arbitrary multivariate polynomial computations, which are of particular importance in learning applications; 2) extending the application of coded computing to secure and private computing; 3) reducing the computation/communication load in distributed computing (and distributed learning) by factors that scale with the problem size, without compromising security and privacy guarantees; and 4) enabling $2.36\times$-$12.65\times$ speedup over the state-of-the-art in distributed least-squares linear regression in cloud networks. 

Secure \textit{multiparty computing} (MPC) and secure/private Machine Learning (e.g., \cite{ben1988completeness, 7958569}%\textbf{Qian: we should put several key references here}
) are also extensively studied topics that address a problem setting similar to LCC. As we elaborate in Section~\ref{sec:prior}, compared with conventional methods in this area (e.g., the celebrated BGW scheme for secure/private MPC~\cite{ben1988completeness}), LCC achieves substantial reduction in the amount of randomness, storage overhead, and computation complexity.

\section{Problem Formulation and Examples}\label{section:formulation}
We consider the problem of evaluating a multivariate polynomial~$f:\bV\to\bU$ over a dataset~$X=(X_1,\ldots,X_K)$,\footnote{We focus on the non-trivial case where $K>0$ and $f$ is not constant.} where~$\bV$ and~$\bU$ are vector spaces of dimensions~$M$ and~$L$, respectively, over the field~$\bF$.
%\footnote{While the results about security hold for any large enough field, privacy is well defined only over finite ones.}~$\bF$.
We assume a distributed computing environment with a master and~$N$ workers (Figure~\ref{fig:sys}), in which the goal is to compute $Y_1\triangleq f(X_1), \ldots, Y_K\triangleq f(X_K)$. We denote the \textit{total degree}\footnote{The \textit{total degree} of a polynomial~$f$ is the maximum among all the total degrees of its monomials. When discussing finite~$\bF$, we resort to the canonical representation of polynomials, in which the individual degree within each term  is no more than $(|\mathbb{F}|-1)$.} of the polynomial~$f$ by $\deg f$.

In this setting each worker has already stored a fraction of the dataset prior to computation, in a possibly coded manner. {Specifically, for~$i\in [N]$ (where~$[N]\triangleq\{1,\ldots,N\}$),  worker~$i$ stores $\tilde{X}_i\triangleq g_i(X_1,\ldots,X_K)$, where~$g_i$ is a (possibly random) function, refered to as the encoding function of that worker.} We restrict our attention to linear encoding schemes\footnote{A formal definition is provided in Section \ref{sec:converses}.}, which guarantee low encoding complexity and simple implementation.

% Upon starting the computation,
Each worker~$i \in[N]$ computes~$\tilde{Y}_i\triangleq f(\tilde{X}_i)$ and returns the result to the master. The master waits for {a subset of fastest workers} and then decodes $Y_1,\ldots,Y_K$. This procedure must satisfy several additional requirements: 
\begin{itemize}[leftmargin=*]
 %   \vspace{-3mm}
    \item  
    \textbf{Resiliency}, i.e., robustness against stragglers. Formally, the master must be able to obtain the correct values of~$Y_1,\ldots,Y_K$ even if up to~$S$ workers fail to respond {(or respond after the master executes the decoding algorithm)}, where~$S$ is the \textit{resiliency parameter} of the system. A scheme that guarantees resiliency against~$S$ stragglers is called \textit{$S$-resilient}. 

%\vspace{-3mm}
\item  %\textit
\textbf{Security}, i.e., robustness against adversaries. That is, the master must be able to obtain correct values of~$Y_1,\ldots,Y_K$ even if up to~$A$ workers return arbitrarily erroneous results, where~$A$ is the \textit{security parameter} of the system. A scheme that guarantees security against~$A$ adversaries is called \textit{$A$-secure}.  

 %   \vspace{-3mm}
\item {\textbf{Privacy}, i.e.,} the workers must remain oblivious to the content of the dataset, even if up to~$T$ of them collude, where~$T$ is the \textit{privacy parameter} of the system. Formally, for every~$\cT\subseteq [N]$ of size at most~$T$, we must have $I(X ; \tilde{X}_\cT)=0$,  %\footnote{Equivalently, equation (\ref{equation:privacy}) requires that $\tilde{X}_\cT$ and $X$ are independent. Under this condition, the input data $X$ still appears uniformly random after the colluding workers learn $\tilde{X}_\cT$, which guarantees the privacy.} 
%\begin{equation}\label{equation:privacy}
%    I(X ; \tilde{X}_\cT)=0
%\end{equation}
where~$I$ is mutual information, $\tilde{X}_\cT$ represents the collection of the encoded dataset stored at the workers in~$\cT$, and~$X$ is seen as chosen uniformly at random.\footnote{Equivalently, it requires that $\tilde{X}_\cT$ and $X$ are independent. Under this condition, the input data $X$ still appears uniformly random after the colluding workers learn $\tilde{X}_\cT$, which guarantees the privacy.} 
A scheme which guarantees privacy against~$T$ colluding workers is called \textit{$T$-private}. \footnote{To guarantee that the privacy requirement is well defined, we assume that $\bF$ and $\bV$ are finite whenever $T>0$.}
\end{itemize}

%    \vspace{-3mm}
{More concretely, given any subset of workers that return the computing results (denoted by $\mathcal{K}$), the master computes $(\hat{Y}_1,...,\hat{Y}_K)=h_{\mathcal{K}}(\{\tilde{Y}_i\}_{i\in\mathcal{K}})$, where each $h_{\mathcal{K}}$ is a deterministic function (or is random but independent of both the encoding functions and input data). We refer to the $h_{\mathcal{K}}$'s as \emph{decoding functions}.}\footnote{{Similar to encoding, we also require the decoding function to have low complexity. When there is no adversary ($A=0$), we restrict our attention to \emph{linear decoding schemes}.}} 
{We say that a scheme is  $S$-resilient, $A$-secure, and $T$-private if the master always returns the correct results (i.e., each $Y_i=\hat{Y}_i$), and all above requirements are satisfied.  }

{Given the above framework, we aim to characterize the region for $(S,A,T)$, such that an $S$-resilient, $A$-secure, and $T$-private scheme can be found, given parameters $N$, $K$, and function $f$, for any sufficiently large field $\bF$.}

This framework encapsulates many computation tasks of interest, which we highlight as follows.

%\begin{motivation}\label{motivation:linear}
\textbf{Linear computation.}
Consider a scenario where the goal is to compute $A\vec{b}$ for some dataset $A=\{A_i\}_{i=1}^K$ and vector~$\vec{b}$, which naturally arises in many machine learning algorithms, such as each iteration of linear regression. Our formulation covers this by letting~$\bV$ be the space of matrices of certain dimensions over~$\bF$, $\bU$ be the space of vectors of a certain length over~$\bF$, $X_i$ be $A_i$, and~$f(X_i)=X_i\cdot \vec{b}$ for all~$i\in[K]$. Coded computing for such linear computations has also been studied in~\cite{dutta2016short,bitar2018minimizing,karakus2017straggler,speeding,wang2018fundamental}.
%\end{motivation}

%\begin{motivation} \label{motivation:bilinear} 
\textbf{Bilinear computation.}
Another computation task of interest is to evaluate element-wise products $\{A_i\cdot B_i\}_{i=1}^{K}$ of two lists of matrices
$\{A_i\}_{i=1}^K$ and~$\{B_i\}_{i=1}^K$. This is the key building block for various algorithms, such as %\sout{\blue{efficient }} 
fast distributed matrix multiplication~\cite{yu2018straggler}. Our formulation covers this by letting~$\bV$ be the space of pairs of two matrices of certain dimensions, $\bU$ be the space of matrices of dimension which equals that of the product of the pairs of matrices, $X_i=(A_i, B_i)$, and~$f(X_i)=A_i\cdot B_i$ for all~$i\in[K]$. 
%\end{motivation}

 %\begin{motivation}\label{motivation:Tensors} 
 \textbf{General tensor algebra.}  
 Beyond bilinear operations, distributed computations of multivariate polynomials of larger degree, such as general tensor algebraic functions (i.e. functions composed of inner products, outer products, and tensor contractions) \cite{renteln2013manifolds}, also arise in practice. A specific example is to compute the coordinate transformation of a third-order tensor field at $K$ locations, where given a list of matrices $\{Q^{(i)}\}_{i=1}^K$ and a list of third-order tensors $\{T^{(i)}\}_{i=1}^K$ with matching dimension on each index, the goal is to compute another list of tensors, denoted by $\{T'^{(i)}\}_{i=1}^K$, of which each entry is defined as $T'^{(i)}_{j'k'\ell'}\triangleq \sum \limits_{j,k,\ell}T'^{(i)}_{jk\ell} Q^{(i)}_{jj'}Q^{(i)}_{kk'}Q^{(i)}_{\ell\ell'}$. Our formulation covers all functions within this class by letting~$\bV$ be the space of input tensors, $\bU$ be the space of output tensors, $X_i$ be the inputs, and $f$ be the tensor function.
 These computations are not studied by state-of-the-art coded computing frameworks. 
%\end{motivation}

%\begin{motivation}\label{motivation:Gradient} 
\textbf{Gradient computation.} Another general class of functions arises from gradient decent algorithms and their variants, which are the workhorse of today's learning tasks~\cite{shalev2014understanding}. 
%Gradient descent, and its variants, are the workhorse of today's learning tasks
% and their distributed computing is essential~\cite{shalev2014understanding}.
The computation task for this class of functions is to consider one iteration of the gradient decent algorithm, and to evaluate the gradient of the empirical risk $\nabla L_\cS(h)\triangleq \avg_{z\in \cS} \nabla \ell_h(z)$, given a hypothesis~$h:\bR^d\to \bR$, a respective loss function~$\ell_h:\bR^{d+1}\to \bR$, and a training set~$\cS\subseteq \bR^{d+1}$, where~$d$ is the number of features.
%In each iteration of these techniques, the master possesses a hypothesis~$h:\bR^p\to \bR$, a respective loss function~$\ell_h:\bR^{p+1}\to \bR$, and a training set~$\cS\subseteq \bR^{p+1}$, where~$p$ is the number of features.
%The master wishes to utilize the servers for the evaluation of $\nabla L_\cS(h)\triangleq \avg_{z\in \cS} \nabla \ell_h(z)$, the gradient of the empirical risk. 
In practice, this computation is carried out by partitioning~$\cS$ into~$K$ subsets~$\{\cS_i\}_{i=1}^K$ of equal sizes, evaluating the partial gradients~$\{\nabla L_{\cS_i}(h)\}_{i=1}^K$ distributedly, and computing the final result using~$\nabla L_\cS(h)=\avg_{i\in[K]}\nabla L_{\cS_i}(h)$. %In many common cases the function~$\ell_h(z)$ is a multivariate polynomial in the entries of~$z$, and thus each entry of~$\nabla L_{\cS_i}(h)$ is a multivariate polynomial in the entries of~$\cS_i$. For example, these cases include the squared loss function whenever~$h$ is a polynomial, or any case where polynomial approximation is used.
%and any case where the $0\mbox{-} 1$ loss is approximated by a Taylor expansion of the error function~$\erf$. 
%This problem fits our framework %\red{in the sense that}
%by letting~$\bV=\bR^{|\cS_i|\times(d+1)}$, $\bU=\bR^r$, where~$r$ is the number of coefficient variables in~$h$, $X_i$ be~$\cS_i$, and~$f(X_i)=\frac{1}{K}\cdot\nabla L_{\cS_i}(h)$ for all~$i\in[K]$. Following the completion of the computation of~$f(X_i)$ for all~$i\in[K]$, the master is only left to compute~$\sum_{i\in[K]}f(X_i)=\nabla L_\cS(h)$. 
We present a specific example of applying this computing model to least-squares regression problems in Section~\ref{sec:regression}.  
%\end{motivation}

\section{Main Results and Prior Works}

%  We summarize our main contributions in the following theorems, and contrast them to previous works {in the following subsections}. 

%\begin{theorem}\label{thm:lcc}
 %   Given a number of workers~$N$, a dataset~$X=(X_1,\ldots,X_K)$, and a polynomial~$f$, the tuple~$(S,A,T)$ is achievable, i.e., there exists an~$S$-resilient, $A$-secure, and~$T$-private scheme for computing~$\{f(X_i)\}_{i=1}^K$ by using~$N$ workers, if $N\ge (K+T-1)\deg f+S+2A+1$.
    % \begin{align*}
    %     N\ge (K+T-1)\deg f+S+2A+1.
    % \end{align*}
    
%      Consider a problem of evaluating any polynomial~$f$, for any values of~$N$ and $K$. LCC provides a $T$-private and $A$-secure computation of~$f$, for any pair $(T,A)$ satisfying
%     \begin{align}
%         N\ge (K+T-1)\deg(f)+2A+1,
%     \end{align}
%   for any field  $\bF_q$ that is sufficiently large (i.e., $q\ge N+K$).  
%\end{theorem}

We now state our main results and discuss their connections with prior works. Our first theorem characterizes the region for $(S,A,T)$ that LCC achieves (i.e., the set of all feasible $S$-resilient, $A$-secure, and $T$-private schemes via LCC as defined in the previos section).

\begin{theorem}\label{thm:lcc}
    Given a number of workers~$N$ and a dataset~$X=(X_1,\ldots,X_K)$, %Lagrange Coded Computing (
    LCC %)
    provides an~$S$-resilient, $A$-secure, and~$T$-private scheme for computing~$\{f(X_i)\}_{i=1}^K$ for any polynomial $f$, as long as
  \begin{align}\label{eq:lccThm}
      (K+T-1)\deg f+S+2A+1\leq N.
  \end{align}  
\end{theorem}

\begin{remark}
To prove Theorem \ref{thm:lcc}, we formally present LCC in Section \ref{section:Lagrange}, which achieves the stated resiliency, security, and privacy. The key idea  is to encode the input dataset using the well-known Lagrange polynomial.
    %, in order to create computation redundancy across the workers.  
    In particular, encoding functions (i.e.,~$g_i$'s) in LCC amount to evaluations of a Lagrange polynomial of degree~$K-1$ at~$N$ distinct points. Hence, computations at the workers amount to evaluations of a \textit{composition} of that polynomial with the desired function~$f$.
    %This composed polynomial is interpolated by the master, that later evaluates it at certain points to finalize the computation. 
    Therefore, inequality~(\ref{eq:lccThm}) may simply be seen as the number of evaluations that are necessary and sufficient in order to interpolate the composed polynomial, which is later evaluated at a certain point to finalize the computation. LCC also has a number of additional properties of interest. First, the proposed encoding is \textit{identical} for all computations $f$, which allows pre-encoding of the data without knowing the identity of the computing task (i.e., universality). Second, decoding and encoding rely on polynomial interpolation and evaluation, and hence efficient off-the-shelf subroutines can be used.\footnote{{A more detailed discussion on the coding complexities of LCC can be found in Appendix \ref{app:comp}. }}
%In LCC, all workers compute the same function~$f$. This suites the common scenario where a function of interest is of the form~$F(X_1,\ldots,X_k)=g(f(X_1),\ldots,f(X_k))$, where~$f$ is a ``hard to compute'' function and~$g$ is an ``easy to compute'' one. This is in accordance with common tasks such as matrix multiplication, MapReduce, and gradient computation, as well as in accordance with the BGW setting.
\end{remark}

\begin{remark}
Besides the coding approach presented to achieve Theorem \ref{thm:lcc}, a variation of LCC can be used to achieve any $(S,A,T)$ as long as $ K(S+2A+\textup{deg}\ f  \cdot T + 1)\leq N$. This scheme (presented in Appendix \ref{app:ulcc}) achieves an improved region when $N<K\textup{deg}f-1$ and $T=0$, where it recovers the \emph{uncoded repetition} scheme. For brevity, we refer the better of these two scheme as LCC when presenting optimality results (i.e., Theorem \ref{thm:opt}). 
\end{remark}

\begin{remark}
    Note that LHS of inequality~(\ref{eq:lccThm}) is independent of the number of workers~$N$, hence the key property of LCC is that adding $1$ worker can increase its resilience to stragglers by 1 or its security to malicious servers by $1/2$, while keeping the privacy constraint $T$ the same. Note that using an uncoded replication based approach, to increase the resiliency to stragglers by 1, one needs to essentially repeat \emph{each} computation once more (i.e., requiring $K$ more machines as opposed to $1$ machine in LCC). This result essentially extends the well-known optimal scaling of error-correcting codes (i.e., adding one parity can provide robustness against one erasure or~$1/2$ error in optimal maximum distance separable codes) to the distributed computing paradigm. 
\end{remark}

Our next theorem demonstrates the optimality of LCC. 
\begin{theorem}\label{thm:opt}
LCC achieves the optimal trade-off between resiliency, security, and privacy (i.e., achieving the largest region of (S,A,T)) for any multilinear function f among all computing schemes that uses linear encoding, for all problem scenarios. 
Moreover, when focusing on the case where no security constraint is imposed, LCC is optimal for any polynomial $f$ among all schemes with additional constraints of linear decoding and sufficiently large (or zero) characteristic of $\bF$.
\end{theorem}
%\begin{theorem}%\label{thm:opt}
%    Lagrange Coded Computing (LCC) \newline
%    (a) provides the optimum resiliency, by tolerating the maximum possible number of stragglers \textup{(}$S$\textup{)}, among all linear schemes (i.e., both encoding and decoding functions are linear), when the characteristic of $\bF$ is sufficiently large;  \newline
%    (b) provides the optimum security, by protecting against the maximum possible number of adversaries \textup{(}$A$\textup{)}, among all schemes that use linear encoding, for any multilinear function $f$;  \newline
%    (c) uses the minimum amount of randomness for universally achieving the same privacy requirement \textup{(}$T$\textup{)} for all linear $f$.
%\end{theorem}
{\begin{remark}
Theorem \ref{thm:opt} is proved in Section~\ref{sec:converses}. The main proof idea is to show that any computing strategy that outperforms LCC %(or its uncoded version)
would violate the decodability requirement, by finding two instances of the computation process where the same intermediate computing results correspond to different output values. %there would be scenarios where all available computing results are degenerated {(i.e., constants)}, while the computing results needed by the master are variable, violating the decodability requirement. }
\end{remark}}
%\begin{remark}
%\textcolor{red}{old thm} 
%\end{remark}

\begin{remark}
 In addition to the result we show in Theorem \ref{thm:opt}, we can also prove that LCC achieves optimality in terms of the amount of randomness used in data encoding. Specifically, we show in Appendix \ref{app:pt_rand} that LCC requires injecting the minimum amount of randomness, among all computing schemes that universally achieve the same resiliency-security-privacy tradeoff for all linear functions $f$.
\end{remark}

%todo new scheme

We conclude this section by discussing several lines of related work in the literature and contrasting them with LCC.

%next review several previous works in coded computing and contrast them with LCC. In addition, our results are in tight connection with secure \textit{multiparty computation} (MPC), specifically with the well known BGW scheme~\cite{ben1988completeness}, on which we elaborate in Subsection~\ref{section:BGW}. 

%\vspace{-3mm}
\subsection{LCC vs. Prior Works}\label{sec:prior}
%\vspace{-3mm}
The study of coding theoretic techniques for accelerating large scale distributed tasks (a.k.a. \emph{coded computing}) was initiated in {\cite{lee2015speeding,LMA_all,Li2018fundamental}}. Following works focused largely on matrix-vector and matrix-matrix multiplication (e.g.,~\cite{dutta2016short, NIPS2017_7027, dutta2018on, yu2018straggler}), gradient computation in gradient descent algorithms (e.g.,~\cite{TandonLDK17,raviv2017gradient,li2017near}), communication reduction via coding (e.g.,~\cite{li2017coded,ezzeldin2017communication,graphlong,2018arXiv180203049K}), and secure and private computing (e.g.,~\cite{MaddahAli,chen2018draco}).

LCC recovers several previously studied results as special cases. For example, setting~$f$ to be the identity function and $\bV=\bU$  reduces to the well-studied case of \textit{distributed storage}, in which Theorem~\ref{thm:lcc} is well known (e.g., the Singleton bound~\cite[Thm.~4.1]{roth2006introduction}). Further, as previously mentioned, $f$ can correspond to matrix-vector and matrix-matrix multiplication, in which the special cases of Theorem~\ref{thm:lcc} are known as well~\cite{speeding,yu2018straggler}.

More importantly, LCC improves and generalizes these works on coded computing in a few aspects: 
\textit{Generality}--LCC significantly generalizes prior works to go beyond linear and bilinear computations that have so far been the main focus in this area, and can be applied to arbitrary multivariate polynomial computations that arise in machine learning applications. In fact, many specific computations considered in the past can be seen as special cases of polynomial computation. This includes matrix-vector multiplication, matrix-matrix multiplication, and gradient computation whenever the loss function at hand is a polynomial, or is approximated by one. \textit{Universality}--once the data has been coded, any polynomial up to a certain degree can be computed distributedly via LCC. In other words, data encoding of LCC can be \emph{universally} used for any polynomial computation. This is in stark contrast to previous task specific coding techniques in the literature. Furthermore, workers apply the same computation as if no coding took place; a feature that reduces computational costs, and prevents ordinary servers from carrying the burden of outliers.
\textit{Security and Privacy}--other than a handful of works discussed above, straggler mitigation (i.e., resiliency) has been the primary focus of the coded computing literature. This work extends the application of coded computing to secure and private computing for general polynomial computations.

%\subsection{LCC vs. MPC} \label{section:BGW}

Providing security and privacy for \textit{multiparty computing} (MPC) and Machine Learning systems is an extensively studied topic which addresses a problem setting similar to LCC.  To illustrate the significant role of LCC in secure and private computing, let us consider the celebrated BGW MPC scheme~\cite{ben1988completeness}. \footnote{Conventionally, the BGW scheme operates in a multi-round fashion, requiring significantly more communication overhead than one-shot approaches. For simplicity of comparison, we present a modified one-shot version of BGW. }

Given inputs~$\{X_i\}_{i=1}^K$, BGW first uses Shamir's scheme \cite{Shamir:1979:SS:359168.359176} to encode the dataset in a privacy-preserving manner as~$P_{i}(z) = X_{i}+Z_{i,1}z+\ldots+Z_{i,T}z^T$ for every~$i\in[K]$, where $Z_{i,j}$'s are i.i.d uniformly random variables and $T$ is the number of colluding workers that should be tolerated. The key distinction between the data encoding of BGW scheme and LCC is that we instead use Lagrange polynomials to encode the data. This results in significant reduction in the amount of randomness needed in data encoding (BGW needs $KT$ $z_{i,j}$'s while as we describe in the next section, LCC only needs $T$ amount of randomness).

The BGW scheme will then store~$\{P_{i}(\alpha_\ell)\}_{i\in [K]}$ to worker~$\ell$ for every~$\ell\in[N]$, given some distinct values $\alpha_1,\ldots,\alpha_N$. The computation is then carried out by evaluating $f$ over \emph{all} stored coded data at the nodes.
In the LCC scheme, on the other hand, each worker $\ell$ only needs to store \emph{one} encoded data ($\tilde{X}_{\ell}$) and  compute $f(\tilde{X}_{\ell})$. This gives rise to the second key advantage of LCC, which is a factor of $K$ in storage overhead and computation complexity at each worker.

After computation, each worker~$\ell$ in the BGW scheme %\blue
{has essentially evaluated the polynomials $\{f(P_{i}(z))\}_{i=1}^K$} at $z=\alpha_\ell$, whose degree is at most~$\deg(f)\cdot T$. Hence, if no straggler or adversary appears (i.e, $S=A=0$), the master can recover all required results $f(P_{i}(0))$'s, through polynomial interpolation, as long as $N\ge \deg(f)\cdot T+1$ workers participated in the computation\footnote{It is also possible to use the conventional multi-round BGW, which only requires $N\geq 2T+1$ workers to ensure $T$-privacy. However, multiple rounds of computation and communication ($\Omega(\log$ deg$(f) )$ rounds) are needed, which further increases its communication overhead.}. Note that under the same condition, LCC scheme requires  $N\ge \deg(f)\cdot (K+T-1)+1$  number of workers, which is larger than that of the BGW scheme. 

Hence, in overall comparison with the BGW scheme, LCC results in a factor of $K$ reduction in the amount of randomness, storage overhead, and computation complexity, while requiring more workers to guarantee the same level of privacy. This is summarized in Table~\ref{table:BGW}.\footnote{A BGW scheme was also proposed in \cite{ben1988completeness} for secure MPC, however for a substantially different setting. Similarly, a comparison can be made by adapting it to our setting, leading to similar results, which we omit for brevity.
%Similarly, comparison can be made by adapting xx to our setting, omitted for breivity.
}

\begin{table}[t]
\centering
\begin{tabular}{l|c|c|}
\hhline{~--}
  & \cellcolor[HTML]{C0C0C0}BGW  & \cellcolor[HTML]{C0C0C0}LCC \\ \hline
\multicolumn{1}{|l|}{\cellcolor[HTML]{C0C0C0}\shortstack[l]{Complexity\\per worker}}               & $K$         & $1$    \\ \hline
\multicolumn{1}{|l|}{\cellcolor[HTML]{C0C0C0}\shortstack[l]{Frac. data\\per worker}}        & $1$                    & $1/K$                          \\ \hline
\multicolumn{1}{|l|}{\cellcolor[HTML]{C0C0C0}Randomness}               & $ KT$                       & $T$                        \\ \hline
\multicolumn{1}{|l|}{\cellcolor[HTML]{C0C0C0}\shortstack[l]{Min. num. \\of workers}} & {\footnotesize$\deg(f)(T+1)$}       & {\footnotesize $\deg(f)(K+T-1)+1$}       \\ \hline
\end{tabular}
\vspace{1mm}
\caption{Comparison between BGW based designs and LCC. The computational complexity is normalized by that of evaluating $f$; randomness, which refers to the number of random entries used in encoding functions, is normalized by the length of $X_i$. 
%Details of the schemes are explained in Section~\ref{section:BGW} and Section~\ref{section:Lagrange}.  
}
\vspace{-7mm}
\label{table:BGW}
\end{table}

%security was studied by a handful of works~\cite{chen2018draco}, and is improved in our construction.
%\textit{Privacy}--
%privacy of specific computation tasks has been studied before (e.g.,~\cite{MaddahAli}). This work improves upon~\cite{MaddahAli} by offering a one-shot (rather than iterative) approach for general computing functions.

%\blue
Recently,~\cite{MaddahAli} has also combined ideas from the BGW scheme  and~\cite{NIPS2017_7027} to form \textit{polynomial sharing}, a {private} coded computation scheme for arbitrary matrix polynomials. However, polynomial sharing inherits the undesired BGW property of performing a communication round for \textit{every} bilinear operation in the polynomial; a feature that drastically increases communication overhead, and is circumvented by the one-shot approach of LCC. \textit{DRACO}~\cite{chen2018draco} is also recently proposed as a secure computation scheme for gradients. Yet, DRACO employs a blackbox approach, i.e., the resulting gradients are encoded rather than the data itself, and the inherent algebraic structure of the gradients is ignored. For this approach, \cite{chen2018draco} shows that a~$2A+1$ \textit{multiplicative} factor of redundant computations is necessary. In LCC however, the blackbox approach is disregarded in favor of an algebraic one, and consequently, a~$2A$ \textit{additive} factor suffices.

LCC has also been recently applied to several applications in which security and privacy in computations are critical. For example, in~\cite{li2018polyshard}, LCC has been applied to enable a scalable and secure approach to sharding in blockchain systems. Also, in~\cite{so2019codedprivateml}, a privacy-preserving approach for machine learning has been developed that leverages LCC to provides substantial speedups over cyrptographic approaches that relay on MPC.

\section{Lagrange Coded Computing}\label{section:Lagrange}
In this Section we prove Theorem~\ref{thm:lcc} by presenting LCC and  characterizing the region for $(S,A,T)$ that it achieves.\footnote{For an algorithmic illustration, see Appendix \ref{app:alg}.} We start with an example to illustrate the key components of LCC.

\subsection{Illustrating Example}\label{section:illustrative}
Consider the function $f(X_i)=X_i^2$, where input $X_i$'s are ~$\sqrt{M}\times \sqrt{M}$ square matrices for some square integer~$M$. We demonstrate LCC in the scenario where the input data~$X$ is partitioned into~$K=2$ batches~$X_1$ and~$X_2$, and the computing system has~$N=8$ workers. In addition, the suggested scheme is %\blue
{$1$-resilient, $1$-secure, and~$1$-private (i.e., achieves  $(S,A,T)=(1,1,1)$)}.

The gist of LCC is picking a uniformly random matrix~$Z$, %of the same dimensions as the~$X_i$'s, 
and encoding~$(X_1,X_2, Z)$ using a Lagrange interpolation polynomial:\footnote{Assume that~$\bF$ is a finite field with~$11$ elements.}
\begin{alignat*}{3}
    u(z)&\triangleq&& X_1\cdot \frac{(z-2)(z-3)}{(1-2)(1-3)}+X_2\cdot \frac{(z-1)(z-3)}{(2-1)(2-3)}+\\
    &~&&Z\cdot \frac{(z-1)(z-2)}{(3-1)(3-2)}.
\end{alignat*}
%and observe that~$u(1)=X_1$ and~$u(2)=X_2$.
We then fix distinct~$\{\alpha_i\}_{i=1}^8$ in~$\bF$ such that~$\{ \alpha_i \}_{i=1}^8\cap [2]=\varnothing$, and let workers~$1,\ldots,8$ store~$u(\alpha_1),\ldots,u(\alpha_8)$. % i.e.,
%\begin{align*}
%    \left( \tilde{X}_1,\ldots,\tilde{X}_8 \right)=(X_1,X_2,Z)\cdot U
%\end{align*}
%where~$U\in\bF_{11}^{3\times 8}$ satisfies~$U_{i,j}=\prod_{\ell\in[3]\setminus \{i\}}\frac{\alpha_j-\ell}{i-\ell}$ for~$(i,j)\in [3]\times [8]$.

First, note that for every~$j\in[8]$, worker~$j$ sees~$\tilde{X}_j$, a linear combination of~$X_1$ and~$X_2$ that is masked by addition of~$\lambda\cdot Z$ for some nonzero~$\lambda\in\bF_{11}$; since~$Z$ is uniformly random, this guarantees perfect privacy for~$T=1$. Next, note that worker~$j$ computes~$f(\tilde{X}_j)=f(u(\alpha_j))$, which is an evaluation of the composition polynomial~$f(u(z))$, whose degree is at most~$4$, at~$\alpha_j$.

Normally, a polynomial of degree $4$ can be interpolated from $5$ evaluations at distinct points. However, the presence of~$A=1$ adversary and~$S=1$ straggler requires the master to employ a Reed-Solomon decoder, and have \textit{three} additional evaluations at distinct points (in general, two additional evaluations for every adversary and one for every straggler). Finally, after decoding polynomial~$f(u(z))$, the master can obtain~$f(X_1)$ and~$f(X_2)$ by evaluating it at~$z=1$ and~$z=2$.

\subsection{General Description}\label{section:GeneralDescription}
Similar to Subsection~\ref{section:illustrative}, we select any~$K+T$ distinct elements~$\beta_1,\ldots,\beta_{K+T}$ from $\bF$, and find a polynomial $u:\bF\rightarrow \bV$ of degree %\blue
{at most} $K+T-1$ such that $u(\beta_i)=X_i$ for any $i\in[K]$, and~$u(\beta_i)=Z_i$ for~$i\in\{K+1,\ldots,K+T\}$, where all~$Z_i$'s are chosen uniformly at random from~$\bV$. This is simply accomplished by letting~$u$ be the \textit{Lagrange interpolation polynomial}
\begin{align*}
    u(z)\triangleq \sum_{j\in[K]}X_j\cdot \prod_{k\in [K+T]\setminus\{j\}}\frac{z-\beta_k}{\beta_j-\beta_k}+\\
    \sum_{j=K+1}^{K+T} Z_j\cdot \prod_{k\in [K+T]\setminus\{j\}}\frac{z-\beta_k}{\beta_j-\beta_k}.
\end{align*}

We then select~$N$ distinct elements~$\{\alpha_i\}_{i\in[N]}$ from $\bF$ such that~$\{\alpha_i\}_{i\in[N]}\cap\{\beta_j\}_{j\in[K]}=\varnothing$ (this requirement is alleviated if~$T=0$), and let $\tilde{X}_i=u(\alpha_i)$ for any $i\in[N]$. That is, the input variables are encoded as \begin{align}\label{equation:perWorkerEncoding}
     \tilde{X}_i\!=\!u(\alpha_i)\!=\!(X_1,\ldots,X_K,Z_{K+1},\ldots,Z_{K+T})\cdot U_i,
 \end{align} 
where~$U\in\bF_q^{(K+T)\times N}$ is the encoding matrix~$U_{i,j}\triangleq\prod_{\ell\in[K+T]\setminus \{i\}}\frac{\alpha_j-\beta_\ell}{\beta_i-\beta_\ell}$, and~$U_i$ is its~$i$'th column.\footnote{By selecting the values of $\alpha_i$'s differently, we can recover the uncoded repetition scheme, see Appendix \ref{app:ulcc}.}

Following the above encoding, each worker~$i$ applies~$f$ on~$\tilde{X}_i$ and sends the result back to the master. Hence, the master obtains~$N-S$ evaluations, at most~$A$ of which are incorrect, of the polynomial~$f(u(z))$. Since $\deg(f(u(z)))\le\deg(f)\cdot (K+T-1)$, and~$N\ge (K+T-1)\deg(f)+S+2A+1$, the master can obtain all coefficients of~$f(u(z))$ by applying Reed-Solomon decoding. Having this polynomial, the master evaluates it at~$\beta_i$ for every~$i\in[K]$ to obtain~$f(u(\beta_i))=f(X_i)$, %\blue
{and hence we have shown that the above scheme is $S$-resilient and $A$-secure.}

% \blue{Now, the servers apply~$f$ on their encoded data, and return~$\{ f(\tilde{X}_i) \}_{i=1}^N=\{ f(u(\alpha_i)) \}_{i=1}^N$ back to the master. Since~$\deg(f(u(z)))=\deg f\cdot (K+T-1)$, it readily follows that the master can obtain all coefficients of~$f$ by using Reed-Solomon decoding, as long as there are at most~$S$ stragglers and~$A$ adversarial workers. Therefore, by evaluating~$f(u(z))$ at~$\beta_1,\ldots,\beta_K$, the proof of the following lemma is immediate.}

% \begin{lemma}\label{lemma:SchemeIsResilientAndSecure}
%     The above scheme is~$S$-resilient and~$A$-secure.
% \end{lemma}

As for the $T$-privacy guarantee of the above scheme, our
% \begin{lemma}\label{lemma:SchemeIsPrivate}
% The above scheme is~$T$-private.
% \end{lemma}
%\begin{proof}
proof relies on the fact that the bottom~$T\times N$ submatrix~$U^{bottom}$ of~$U$ is an MDS matrix (i.e., every~$T\times T$ submatrix of~$U^{bottom}$ is invertible, see Lemma~\ref{lemma:Uproperties} in the supplementary material). Hence, for a colluding set of workers~$\cT\subseteq [N]$ of size~$T$, their encoded data~$\tilde{X}_\cT$ satisfies~$\tilde{X}_\cT=XU_\cT^{top}+ZU_\cT^{bottom}$, where~$Z\triangleq(Z_{K+1},\ldots,Z_{K+T})$, and~$U_\cT^{top}\in \bF_q^{K\times T}$, $U_\cT^{bottom}\in \bF_q^{T\times T}$ are the top and bottom submatrices which correspond to the columns in~$U$ that are indexed by~$\cT$. 
Now, the fact that any~$U_{\cT}^{bottom}$ is invertible implies that the random padding added for these colluding workers is uniformly random, which completely masks the coded data~$XU_\cT^{top}$.  This directly guarantees~$T$-privacy.
%\end{proof}

% \begin{remark}
% This construction is applicable over every finite field with~$q\ge K+T$. Moreover, disregarding the privacy constraint (i.e., setting~$T=0$) provides an~$A$-secure scheme over infinite fields as well.
% \end{remark}

%\textcolor{red}{IF SPACE PERMITS} - A remark about encoding/decoding complexities.

%TODO
%\Qian{ \begin{remark}
% In terms of coding complexities, the decoding process of LCC is essentially computing the Lagrange interpolation for the polynomial $f \circ u$ at $K$ points. This interpolation can be efficiently computed with an almost linear complexity (i.e., $O(K^*\log^2K^*\log \log K^*)$ linear operations in $\bU$), using fast polynomial arithmetic algorithms \cite{kedlaya2011fast}. Furthermore, this complexity can be reduced by simply swapping in any faster interpolation algorithm or Reed-Solomon decoding algorithm.
% \end{remark}}

%\begin{figure}[h!]
%\centering
%\includegraphics[scale=1.7]{file.xx}
%\caption{The Universe}
%\label{fig:universe}
%\end{figure}

\section{Optimality of LCC}\label{sec:converses}

 In this section, we provide a layout for the proof of optimality for LCC (i.e., Theorem \ref{thm:opt}). Formally, we define 
 that a \emph{linear encoding function} is one that computes a linear combination of the input variables (and possibly a list of independent uniformly random keys when privacy is taken into account\footnote{This is well defined as we assumed that $\bV$ is finite when $T>0$.}); while a \emph{linear decoding function} computes a linear combination of workers' output. We essentially need to prove that (a) given any multilinear $f$, any linear encoding scheme that achieves any $(S,A,T)$ requires at least $N\geq(K+T-1)\deg f+S+2A+1$ workers when $T>0$ or $N\ge K\textup{deg}\ f -1$, and $N\geq K(S+2A+1)$ workers in other cases; (b) for a general polynomial $f$, any scheme that uses linear encoding and decoding requires at least the same number of workers, if the characteristic of $\bF$ is $0$ or greater than $\textup{deg}\ f$.
 
 The proof rely on the following key lemma, which characterizes the \emph{recovery threshold} of any encoding scheme, defined as the minimum number of workers that the master needs to wait to guarantee decodability. 
 \begin{lemma}\label{lemma:rec}
 Given any  multilinear $f$, the recovery threshold of any valid linear encoding scheme, denoted by $R$, satisfies 
 \begin{align}\label{ineq:rec}
     R\geq &R_{\textup{LCC}}(N,K,f)\triangleq \nonumber \\ &\min \{(K-1)\deg f+1, \  N-\lfloor N/K\rfloor+1\}.
 \end{align}
 Moreover, if the encoding scheme is $T$ private, we have $R\geq R_\textup{LCC}(N,K,f)+T\cdot \deg f$.
  \end{lemma}
 
  The proof of Lemma \ref{lemma:rec} can be found in Appendix \ref{pl:rec}, by constructing instances of the computation process for any assumed scheme that achieves smaller recovery threshold, and proving that such scheme fails to achieve decodability in these instances.
 Intuitively, note that the recovery threshold is exactly the difference between $N$ and the number of stragglers that can be tolerated, inequality (\ref{ineq:rec}) in fact proves that LCC (described in Section \ref{section:Lagrange} and Appendix \ref{pl:resiliency}) achieves the optimum resiliency, as it exactly achieves the stated recovery threshold. Similarly, one can verify that Lemma \ref{lemma:rec} essentially states that LCC achieves the optimal tradeoff between resiliency and privacy. 
  
  Assuming the correctness of Lemma \ref{lemma:rec}, the two parts of Theorem \ref{thm:opt} can be proved as follows.
  To prove part (a) of the converses, we need to extend  Lemma \ref{lemma:rec} to also take adversaries into account. This is achieved by using an extended concept of Hamming distance, defined in \cite{yu2018straggler} for coded computing. Part (b) requires generalizing  Lemma \ref{lemma:rec} to arbitrary polynomial functions, which is proved by 
  showing that for any $f$ that achieves any $(S,T)$ pair, there exists a multilinear function with the same degree for which a computation scheme can be found to achieves the same requirement.
   The detailed proofs can be found in Appendices \ref{pl:security} and \ref{pl:resiliency} respectively.

  %Theorem \ref{thm:opt},  can be proved as follows 

% Specifically, given any encoding scheme $\boldsymbol{g}\triangeq \{g_1,...,g_N\}$,  defined as the minimum integer   
 %\textcolor{red}{key lemma} recovery threshold,moreover if T, blah

% \begin{lemma}[Optimal Resiliency]\label{lemma:resiliency}
% Any computing scheme where both the encoding and decoding functions are linear can tolerate at most $S=N-(K-1)\textup{deg}\ f-1$ stragglers when $N\ge K\textup{deg}\ f -1$, and $S=\lfloor N/K \rfloor -1$ stragglers when $N< K\textup{deg}\ f-1$, assuming the characteristic of $\bF_q$ is $0$ or greater than $\textup{deg}\ f$. 
% \end{lemma}

 %Theorem \ref{lemma:resiliency} is proved in Appendix \ref{pl:resiliency}, 

%\section{Experimental Results}
\section{Application to Linear Regression and Experiments on AWS EC2}
\label{sec:regression}

In this section we demonstrate a practical application of LCC in accelerating distributed linear regression, whose gradient computation is a quadratic function of the input dataset, hence matching well the LCC framework. We also experimentally demonstrate its performance gain over state of the arts via experiments on AWS EC2 clusters.

% \subsection{Applying LCC for linear regression}\label{sec:regressionFormulation}
{\bf Applying LCC for linear regression.} Given a feature matrix $\vct{X} \in \mathbb{R}^{m\times d}$ containing $m$ data points of $d$ features, and a label vector $\vct{y} \in \mathbb{R}^m$, a linear regression problem aims to find the weight vector $\vct{w} \in \mathbb{R}^d$ that minimizes the loss $||\vct{X} \vct{w} - \vct{y}||^2$. Gradient descent (GD) solves this problem by iteratively moving the weight along the negative gradient direction, which is in iteration-$t$ computed as $2 \vct{X}^\top(\vct{X}\vct{w}^{(t)}-\vct{y})$.

To run GD distributedly over a system comprising  a master node and $n$ worker nodes, we first partition $\vct{X} = [\vct{X}_1 \cdots  \vct{X}_{n}]^\top$ into $n$ sub-matrices. Each worker stores $r$ coded sub-matrices generated from linearly combining $\vct{X}_j$s, for some parameter $1 \leq r \leq n$. Given the current weight $\vct{w}$, each worker performs computation using its local storage, and sends the result to the master. Master recovers $\vct{X}^\top\vct{X}\vct{w} = \sum_{j=1}^{n} \vct{X}_j \vct{X}_j^\top \vct{w}$ using the results from a subset of fastest workers.\footnote{Since the value of $\vct{X}^\top \vct{y}$ does not vary across iterations, it only needs to be computed once. We assume that it is available at the master for weight updates.
} To measure performance of any linear regression scheme, we consider the metric \emph{recovery threshold} (denoted by $R$), defined as the minimum number of workers the master needs to wait for, to guarantee decodability (i.e., tolerating the remaining stragglers).

We cast this gradient computation to the computing model in Section~\ref{section:formulation}, by grouping the sub-matrices into $K \!\!=\!\! \lceil\frac{n}{r}\rceil$ blocks such that $\vct{X} = [\bar{\vct{X}}_1 \cdots  \bar{\vct{X}}_{K}]^\top$. Then computing $\vct{X} \vct{X}^\top \vct{w}$ reduces to computing the sum of a degree-$2$ polynomial $f(\bar{\vct{X}}_k) = \bar{\vct{X}}_k\bar{\vct{X}}_k^\top \vct{w}$, evaluated over $\bar{\vct{X}}_1,\ldots, \bar{\vct{X}}_K$. Now, we can use LCC to decide on the coded storage as in (\ref{equation:perWorkerEncoding}), and achieve a recovery threshold of $R_{\textup{LCC}} = 2(K-1)+1=2\lceil \tfrac{n}{r}\rceil -1$ (Theorem~\ref{thm:lcc}).\footnote{This recovery threshold is also optimum within a factor of $2$, as we proved in Appendix~\ref{sec:regression_converse}.}

\noindent {\bf Comparisons with state of the arts.} The conventional uncoded scheme picks $r=1$, and has each worker $j$ compute $\vct{X}_j\vct{X}_j^\top\vct{w}$. Master needs result from each work, yielding a recovery threshold of $R_{\textup{uncoded}}=n$. By redundantly storing/processing $r>1$ \emph{uncoded} sub-matrices at each worker, the ``gradient coding'' (GC) methods~\cite{TandonLDK17,Halbawi,raviv2017gradient} code across partial gradients computed from uncoded data, and reduce the recovery threshold to $R_{\textup{GC}} = n-r+1$.
% each $\vct{X}_j$ for $r>1$ times, the ``gradient coding'' (GC) methods~\cite{TandonLDK17,Halbawi,raviv2017gradient} code across partial gradients computed from \emph{uncoded} data at each worker, such that the total gradient can be recovery at the master. For linear regression, each worker stores $r$ uncoded sub-matrices, and GC achieves a recovery threshold of $R_{\textup{GC}} = n-r+1$. 
An alternative ``matrix-vector multiplication based'' (MVM) approach~\cite{lee2015speeding} requires two rounds of computation. In the first round, an intermediate vector $\vct{z} = \vct{X}\vct{w}$ is computed distributedly,
%and decoded at the master, 
which is re-distributed to the workers in the second round for them to collaboratively compute $\vct{X}^\top\vct{z}$. %To mitigate stragglers in each round, 
Each worker stores coded data generated using MDS codes from $\vct{X}$ and $\vct{X}^\top$ respectively. 
%Applying MDS codes to the feature matrix allows the
MVM achieves a recovery threshold of $R_{\textup{MVM}} = \lceil \frac{2n}{r} \rceil$ in \emph{each} round, when the storage is evenly split between rounds. %\textcolor{red}{Finally, another line of works~\cite{karakus2017straggler,maityrobust} further reduce the recovery threshold at the cost of the optimality of the trained mode (e.g., not requiring exact recovery of the gradient). We stick to the exact computation, and compare LCC with GC and MVM.}
% using the cyclic repetition scheme in~\cite{TandonLDK17}, each worker $j$ stores uncoded sub-matrices $\vct{X}_{j}, \ldots, \vct{X}_{j+r-1}$, and sends a liner combination of the results $\vct{X}_{j}\vct{X}_{j}^\top\vct{w}, \ldots, \vct{X}_{j+r-1}\vct{X}_{j+r-1}^\top\vct{w}$. Master recovers the final result $\vct{X}_{1}\vct{X}_{1}^\top\vct{w}+\cdots+\vct{X}_{n}\vct{X}_{n}^\top\vct{w}$ by linearly combining the messages received from any subsets of $n-r+1$ workers, yielding a recovery threshold of $K_{\textup{GC}} = n-r+1$. 
% \noindent {\bf Matrix-vector multiplicaiton based approaches.} An alternative approach to compute  $\vct{X}^\top \vct{X}\vct{w}$ was proposed in~\cite{lee2015speeding}, which we call the ``matrix-vector multiplication based'' (MVM) approach. MVM consists of two rounds. In the first round, an intermediate vector $\vct{z} = \vct{X}\vct{w}$ is computed distributedly and decoded at the master, which is re-distributed to the workers in the second round for them to collaboratively compute the final result $\vct{X}^\top\vct{z}$. To mitigate stragglers, each worker splits its memory and store coded data generated from $\vct{X}$ and $\vct{X}^\top$ respectively. Applying MDS codes (e.g., Reed-Solomon code) to the feature matrix allows the MVM scheme to achieve a recovery threshold of $K_{\textup{MVM}} = \lceil \frac{2n}{r} \rceil$ in \emph{each} round, when the storage is evenly split between two rounds.

% Compared with GC, LCC reduces the recovery threshold by about $r/2$ times, by coding directly on data. 
Compared with GC, LCC codes directly on data, and reduces the recovery threshold by about $r/2$ times. While the amount of computation and communication at each worker is the same for GC and LCC, LCC is expected to finish much faster due to its much smaller recovery threshold. %However, GC schemes are applicable for general learning problems where the gradient can be arbitrary functions of the data.
Compared with MVM, LCC achieves a smaller recovery threshold than that in each round of MVM (assuming even storage split).
% only one round of computation and communication, with a smaller recovery threshold than that of MVM in each round (assuming even storage split).
% However, MVM requires less amount of computation at each worker than LCC. While LCC has each worker send a dimension-$d$ vector in each iteration, each MVM worker sends two vectors whose sizes are respectively proportional to $m$ and $d$.
While each MVM worker performs less computation in each iteration, it sends two vectors whose sizes are respectively proportional to $m$ and $d$, whereas each LCC worker only sends one dimension-$d$ vector.

We run linear regression on AWS EC2 using Nesterov's accelerated gradient descent, where all nodes are implemented on \texttt{t2.micro} instances. %and compare the performance of LCC, GC~\cite{TandonLDK17} (the cyclic repetition scheme), and MVM~\cite{lee2015speeding}. 
% The master and the workers are implemented on \texttt{t2.micro} instances using {\tt Python}. 
We generate synthetic datasets of $m$ data points, by 1) randomly sampling a true weight $\vct{w}^{*}$, 2) randomly sampling each input $\vct{x}_i$ of $d$ features and computing its output $y_i=\vct{x}_i^\top\vct{w}^{*}$. For each dataset, we run GD for $100$ iterations over $n=40$ workers.
We consider different dimensions of input matrix $\vct{X}$ as listed in the following scenarios.
% by varying $m$ and $d$, and un GD for $100$ iterations over $n=40$ workers.
% as in the three scenarios listed below. For each $(m,d)$ pair, we run GD for $100$ iterations over $n=40$ workers. %To better compare the effectiveness of straggler mitigation using different schemes,
%\vspace{-3mm}
\begin{itemize}[leftmargin=*]
    \item Scenario 1 \& 2: $(m,d)=(8000,7000)$.
  %  \vspace{-2mm}
    \item Scenario 3: $(m,d)=(160000,500)$.
    %$\vct{X} \in \mathbb{R}^{8000 \times }$
\end{itemize}
%\vspace{-3mm}

We let the system run with naturally occurring stragglers in scenario 1. To mimic the effect of slow/failed workers, we artificially introduce stragglers in scenarios 2 and 3, by imposing a $0.5$ seconds delay on each worker with probability $5\%$ in each iteration.

To implement LCC, we set the $\beta_i$ parameters to $1,...,\frac{n}{r}$, and the $\alpha_i$ parameters to $0,\ldots,n-1$. To avoid numerical instability due to large entries of the decoding matrix, 
% we can first quantize input data with finite precision, embed them into a finite field that covers the range of possible output data, and apply LCC in the finite field with exact computations. However in all of our experiments the gradients were calculated correctly without carrying out this step.
we can embed input data into a large finite field, and apply LCC in it with exact computations. However in all of our experiments the gradients are calculated correctly without carrying out this step.

\pgfplotsset{compat=1.11,
    /pgfplots/ybar legend/.style={
    /pgfplots/legend image code/.code={%
       \draw[##1,/tikz/.cd,yshift=-0.28em]
        (0cm,0cm) rectangle (3pt,0.8em);},
   },
}

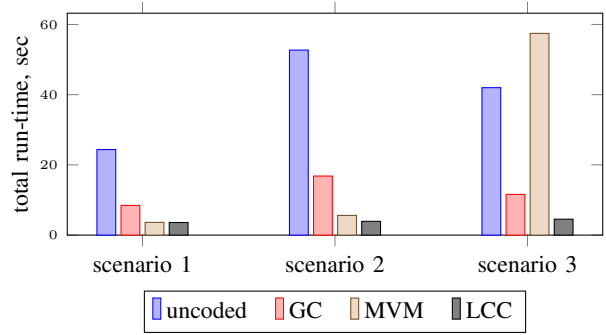
\begin{figure}[t]
\centering
\begin{tikzpicture}{center} 
\begin{axis}[
    ybar,
    %area legend,
    %ymin=1
    %x=8cm,
    yticklabel style = {font=\tiny},
    bar width=.25cm,
    width=8.7cm,
    height=.25\textwidth,
    enlarge x limits={abs=1cm},
    %enlargelimits=0.15,
    %label y style={font=\large}
    %enlarge x limits  = 1
    legend style={at={(0.5,-0.25)},
      anchor=north,legend columns=-1,/tikz/every even column/.append style={column sep=0.3cm}},
    ylabel={total run-time, sec},style = {font=\tiny},
    y label style={at={(axis description cs:-0.05,.5)},anchor=south},
    %symbolic x coords={scenario one,scenario two,scenario three,scenario four},
    %label x style={font=\large},
    %tick label x style={font=\large},
    xticklabels={scenario 1,scenario 2,scenario 3}, style = {font=\small},
    xtick=data,
    ymin=0,
    %nodes near coords,
    %nodes near coords align={vertical},
    ]
\addplot coordinates {(1,24.362) (1.5,52.700)(2,41.994)};
\addplot coordinates {(1,8.464) (1.5,16.821)(2,11.589)};
\addplot coordinates {(1,3.626) (1.5,5.607)(2,57.457)};
\addplot coordinates {(1,3.587) (1.5,3.925)(2,4.541)};
\legend{uncoded, GC, MVM, LCC}
\end{axis}
\end{tikzpicture}
%\vspace{-4mm}
\caption{{\small Run-time comparison of LCC with other three schemes: conventional uncoded, GC, %(cyclic repetition scheme in~\cite{TandonLDK17}), 
and MVM. %Scenarios 1 and 2 have a $8000 \times 7000$ (near) square input, while scenario 3 has a $ 160000 \times 500$ tall input. In scenarios 2 and 3, we add artificial stragglers. In all scenarios, we optimize the run-times of all schemes over the storage parameter $r$.
}}
%\vspace{-4mm}
\label{fig:run-time}
\end{figure}

% \vspace{-3mm}
% \subsection{Results}
% \vspace{-3mm}
{\bf Results.} %In this section we present the results of our numerical comparisons with uncoded, GC and MVM schemes. 
%We first note that for the uncoded scheme, each worker stores and processes one data batch ($r=1$). 
% For the GC and LCC schemes, we select the optimal $r$ subject to the memory size of the \texttt{t2.micro} instance to minimize the total run-time. For MVM, we further optimized the run-time over the computation/storage assigned between two rounds of matrix-vector multiplications.
% %For the uncoded, GC, and PCR schemes that have fixed recovery thresholds across iterations, 
% We plot the run-time performance in all four scenarios in Figure~\ref{fig:run-time}, and also list the breakdowns of their run-times in Tables~\ref{table:scenario one} to \ref{table:scenario four}. The computation (comp.) time was measured as the summation of the maximum local processing time among all non-straggling workers, over 100 iterations. The communication (comm.) time is computed as the difference between the total run-time and the computation time. %Finally, we plot the CDFs of the per iteration run-time for the PCR and BCC schemes in the four scenarios in Figure~\ref{fig:cdf}. 
For GC and LCC, we optimize the total run-time over $r$ subject to local memory size. % of \texttt{t2.micro} instance. 
For MVM, we further optimize the run-time over the storage assigned between two rounds of matrix-vector multiplications.
%For the uncoded, GC, and PCR schemes that have fixed recovery thresholds across iterations, 
We plot the measured run-times in Figure~\ref{fig:run-time}, and list the detailed breakdowns of all scenarios in Appendix~\ref{sec:experiments}. 

We draw the following conclusions from experiments.
%\vspace{-3mm}
\begin{itemize}[leftmargin=*]
\item LCC achieves the least run-time in all scenarios. In particular, LCC speeds up the uncoded scheme by $6.79\times$-$13.43\times$, the GC scheme by $2.36$-$4.29\times$, and the MVM scheme by $1.01$-$12.65\times$.
% significantly improves upon MVM for a tall feature matrix, i.e.,~$m$ sufficiently larger than $d$, by a %$12.65\times \sim 141.19\times$ 
% factor of $12.65\times$.

%\vspace{-3mm}
\item In scenarios 1 \& 2 where the number of inputs $m$ is close to the number of features $d$, LCC achieves a similar performance as MVM. However, when we have much more data points in scenario 3, LCC finishes substantially faster than MVM by as much as $12.65\times$.
% 1 that contains naturally occurring stragglers, LCC improves upon GC by a $2.36 \times$ factor when the number of inputs $m$ is close to the number of features $d$ (scenario 1). In this case LCC has a similar performance to MVM. However, when we have much more data points in scenario 3, LCC substantially reduces the run-time of MVM by as much as $141.19 \times$. 
The main reason for this subpar performance is that MVM requires large amounts of data transfer from workers to the master in the first round and from master to workers in the second round (both are proportional to $m$). However, the amount of communication from each worker or master is  proportional to $d$ for all other schemes, which is much smaller than $m$ in scenario 3. %In scenarios 2 \& 4 that contain artificial stragglers, LCC improves the run-time of GC by $2.55\times \sim 4.29\times$, and the run-time of MVM by $1.43\times \sim 12.65\times$.

\end{itemize}

\Removed{\section{Conclusion}
\input{Conclusion.tex}}

	\section*{Acknowledgement}
	This material is based upon work supported by Defense Advanced Research Projects Agency (DARPA) under Contract No. HR001117C0053, ARO award W911NF1810400, NSF grants CCF-1703575, ONR Award No. N00014-16-1-2189, and CCF-1763673. The views, opinions, and/or findings expressed are those of the author(s) and should not be interpreted as representing the official views or policies of the Department of Defense or the U.S. Government. M. Soltanolkotabi is supported by the Packard Fellowship in Science and Engineering, a Sloan Research Fellowship in Mathematics, an NSF-CAREER under award \#1846369, the Air Force Office of Scientific Research Young Investigator Program (AFOSR-YIP) under award \#FA9550-18-1-0078, an NSF-CIF award \#1813877, and a Google faculty research award. Qian Yu is supported by the Google PhD Fellowship.
	
\bibliography{main}

\begin{thebibliography}{10}

\bibitem{abadi2016tensorflow}
M.~Abadi, P.~Barham, J.~Chen, Z.~Chen, A.~Davis, J.~Dean, M.~Devin,
  S.~Ghemawat, G.~Irving, M.~Isard, {\em et~al.}, ``Tensorflow: A system for
  large-scale machine learning.,'' in {\em OSDI}, vol.~16, pp.~265--283, 2016.

\bibitem{dean2013tail}
J.~Dean and L.~A. Barroso, ``The tail at scale,'' {\em Communications of the
  ACM}, vol.~56, no.~2, pp.~74--80, 2013.

\bibitem{LiMu}
M.~Li, D.~G. Andersen, A.~Smola, and K.~Yu, ``Communication efficient
  distributed machine learning with the parameter server,'' in {\em Proceedings
  of the 27th International Conference on Neural Information Processing Systems
  - Volume 1}, NIPS'14, (Cambridge, MA, USA), pp.~19--27, MIT Press, 2014.

\bibitem{MultiTask}
N.~J. Yadwadkar, B.~Hariharan, J.~E. Gonzalez, and R.~Katz, ``Multi-task
  learning for straggler avoiding predictive job scheduling,'' {\em Journal of
  Machine Learning Research}, vol.~17, no.~106, pp.~1--37, 2016.

\bibitem{li2014communication}
M.~Li, D.~G. Andersen, A.~J. Smola, and K.~Yu, ``Communication efficient
  distributed machine learning with the parameter server,'' in {\em Advances in
  Neural Information Processing Systems}, pp.~19--27, 2014.

\bibitem{oneAdvarsarialStraggler}
P.~Blanchard, R.~Guerraoui, J.~Stainer, {\em et~al.}, ``Machine learning with
  adversaries: Byzantine tolerant gradient descent,'' in {\em Advances in
  Neural Information Processing Systems}, pp.~118--128, 2017.

\bibitem{CramerMPC}
R.~Cramer, I.~B. Damgrd, and J.~B. Nielsen, {\em Secure Multiparty Computation
  and Secret Sharing}.
\newblock New York, NY, USA: Cambridge University Press, 1st~ed., 2015.

\bibitem{ShareMind}
D.~Bogdanov, S.~Laur, and J.~Willemson, ``Sharemind: A framework for fast
  privacy-preserving computations,'' in {\em Proceedings of the 13th European
  Symposium on Research in Computer Security: Computer Security}, ESORICS '08,
  (Berlin, Heidelberg), pp.~192--206, Springer-Verlag, 2008.

\bibitem{speeding}
K.~Lee, M.~Lam, R.~Pedarsani, D.~Papailiopoulos, and K.~Ramchandran, ``Speeding
  up distributed machine learning using codes,'' {\em IEEE Transactions on
  Information Theory}, vol.~64, pp.~1514--1529, March 2018.

\bibitem{TandonLDK17}
R.~Tandon, Q.~Lei, A.~G. Dimakis, and N.~Karampatziakis, ``Gradient coding:
  Avoiding stragglers in distributed learning,'' in {\em Proceedings of the
  34th International Conference on Machine Learning, {ICML} 2017, Sydney, NSW,
  Australia, 6-11 August 2017}, pp.~3368--3376, 2017.

\bibitem{maityrobust}
R.~K. Maity, A.~S. Rawat, and A.~Mazumdar, ``Robust gradient descent via moment
  encoding with ldpc codes,'' {\em SysML Conference}, 2018.

\bibitem{karakus2017straggler}
C.~Karakus, Y.~Sun, S.~Diggavi, and W.~Yin, ``Straggler mitigation in
  distributed optimization through data encoding,'' in {\em Advances in Neural
  Information Processing Systems}, pp.~5440--5448, 2017.

\bibitem{li2017near}
S.~Li, S.~M.~M. Kalan, A.~S. Avestimehr, and M.~Soltanolkotabi, ``Near-optimal
  straggler mitigation for distributed gradient methods,'' {\em arXiv preprint
  arXiv:1710.09990}, 2017.

\bibitem{Halbawi}
W.~Halbawi, N.~A. Ruhi, F.~Salehi, and B.~Hassibi, ``Improving distributed
  gradient descent using reed-solomon codes,'' {\em CoRR}, vol.~abs/1706.05436,
  2017.

\bibitem{raviv2017gradient}
N.~Raviv, I.~Tamo, R.~Tandon, and A.~G. Dimakis, ``Gradient coding from cyclic
  mds codes and expander graphs,'' {\em arXiv preprint arXiv:1707.03858}, 2017.

\bibitem{ye2018communication}
M.~Ye and E.~Abbe, ``Communication-computation efficient gradient coding,''
  {\em arXiv preprint arXiv:1802.03475}, 2018.

\bibitem{lee2015speeding}
K.~Lee, M.~Lam, R.~Pedarsani, D.~Papailiopoulos, and K.~Ramchandran, ``Speeding
  up distributed machine learning using codes,'' {\em NIPS Workshop on Machine
  Learning Systems}, Dec. 2015.

\bibitem{LMA_all}
S.~Li, M.~A. Maddah-Ali, and A.~S. Avestimehr, ``Coded {M}ap{R}educe,'' in {\em
  Proceedings of the 2015 53rd Annual Allerton Conference on Communication,
  Control, and Computing (Allerton)}, pp.~964--971, Sept. 2015.

\bibitem{yu2017optimally}
Q.~Yu, S.~Li, M.~A. Maddah-Ali, and A.~S. Avestimehr, ``How to optimally
  allocate resources for coded distributed computing?,'' in {\em 2017 IEEE
  International Conference on Communications (ICC)}, pp.~1--7, May 2017.

\bibitem{Li2018fundamental}
S.~Li, M.~A. Maddah-Ali, Q.~Yu, and A.~S. Avestimehr, ``A fundamental tradeoff
  between computation and communication in distributed computing,'' {\em IEEE
  Transactions on Information Theory}, vol.~64, no.~1, pp.~109--128, 2018.

\bibitem{dutta2016short}
S.~Dutta, V.~Cadambe, and P.~Grover, ``Short-dot: Computing large linear
  transforms distributedly using coded short dot products,'' in {\em Advances
  In Neural Information Processing Systems}, pp.~2092--2100, 2016.

\bibitem{NIPS2017_7027}
Q.~Yu, M.~Maddah-Ali, and S.~Avestimehr, ``Polynomial codes: an optimal design
  for high-dimensional coded matrix multiplication,'' in {\em Advances in
  Neural Information Processing Systems 30}, pp.~4406--4416, Curran Associates,
  Inc., 2017.

\bibitem{dutta2018on}
S.~Dutta, M.~Fahim, F.~Haddadpour, H.~Jeong, V.~R. Cadambe, and P.~Grover, ``On
  the optimal recovery threshold of coded matrix multiplication,'' {\em arXiv
  preprint arXiv:1801.10292}, 2018.

\bibitem{MaddahAli}
H.~A. Nodehi and M.~A. Maddah-Ali, ``Limited-sharing multi-party computation
  for massive matrix operations,'' in {\em 2018 IEEE International Symposium on
  Information Theory (ISIT)}, pp.~1231--1235, June 2018.

\bibitem{chen2018draco}
L.~Chen, Z.~Charles, D.~Papailiopoulos, {\em et~al.}, ``Draco: Robust
  distributed training via redundant gradients,'' {\em arXiv preprint
  arXiv:1803.09877}, 2018.

\bibitem{ben1988completeness}
M.~Ben-Or, S.~Goldwasser, and A.~Wigderson, ``Completeness theorems for
  non-cryptographic fault-tolerant distributed computation,'' in {\em
  Proceedings of the twentieth annual ACM symposium on Theory of computing},
  pp.~1--10, ACM, 1988.

\bibitem{7958569}
P.~Mohassel and Y.~Zhang, ``Secureml: A system for scalable privacy-preserving
  machine learning,'' in {\em 2017 IEEE Symposium on Security and Privacy
  (SP)}, vol.~00, pp.~19--38, May 2017.

\bibitem{bitar2018minimizing}
R.~Bitar, P.~Parag, and S.~E. Rouayheb, ``Minimizing latency for secure coded
  computing using secret sharing via staircase codes,'' {\em arXiv preprint
  arXiv:1802.02640}, 2018.

\bibitem{wang2018fundamental}
S.~Wang, J.~Liu, N.~Shroff, and P.~Yang, ``Fundamental limits of coded linear
  transform,'' {\em arXiv preprint arXiv:1804.09791}, 2018.

\bibitem{yu2018straggler}
Q.~Yu, M.~A. Maddah{-}Ali, and A.~S. Avestimehr, ``Straggler mitigation in
  distributed matrix multiplication: Fundamental limits and optimal coding,''
  {\em arXiv preprint arXiv:1801.07487}, 2018.

\bibitem{renteln2013manifolds}
P.~Renteln, {\em Manifolds, Tensors, and Forms: An Introduction for
  Mathematicians and Physicists}.
\newblock Cambridge University Press, 2013.

\bibitem{shalev2014understanding}
S.~Shalev-Shwartz and S.~Ben-David, {\em Understanding machine learning: From
  theory to algorithms}.
\newblock Cambridge university press, 2014.

\bibitem{li2017coded}
S.~Li, S.~Supittayapornpong, M.~A. Maddah-Ali, and S.~Avestimehr, ``Coded
  terasort,'' {\em IPDPSW}, 2017.

\bibitem{ezzeldin2017communication}
Y.~H. Ezzeldin, M.~Karmoose, and C.~Fragouli, ``Communication vs distributed
  computation: an alternative trade-off curve,'' {\em arXiv preprint
  arXiv:1705.08966}, 2017.

\bibitem{graphlong}
S.~Prakash, A.~Reisizadeh, R.~Pedarsani, and S.~Avestimehr, ``Coded computing
  for distributed graph analytics,'' {\em arXiv preprint arXiv:1801.05522},
  2018.

\bibitem{2018arXiv180203049K}
K.~{Konstantinidis} and A.~{Ramamoorthy}, ``{Leveraging Coding Techniques for
  Speeding up Distributed Computing},'' {\em ArXiv e-prints}, 2018.

\bibitem{roth2006introduction}
R.~Roth, {\em Introduction to coding theory}.
\newblock Cambridge University Press, 2006.

\bibitem{Shamir:1979:SS:359168.359176}
A.~Shamir, ``How to share a secret,'' {\em Commun. ACM}, vol.~22, pp.~612--613,
  Nov. 1979.

\bibitem{li2018polyshard}
S.~Li, M.~Yu, S.~Avestimehr, S.~Kannan, and P.~Viswanath, ``Polyshard: Coded
  sharding achieves linearly scaling efficiency and security simultaneously,''
  {\em arXiv preprint arXiv:1809.10361}, 2018.

\bibitem{so2019codedprivateml}
J.~So, B.~Guler, A.~S. Avestimehr, and P.~Mohassel, ``Codedprivateml: A fast
  and privacy-preserving framework for distributed machine learning,'' {\em
  arXiv preprint arXiv:1902.00641}, 2019.

\bibitem{kedlaya2011fast}
K.~S. Kedlaya and C.~Umans, ``Fast polynomial factorization and modular
  composition,'' {\em SIAM Journal on Computing}, vol.~40, no.~6,
  pp.~1767--1802, 2011.

\bibitem{Berlekamp:2006:NBD:2263266.2267601}
E.~{Berlekamp}, ``Nonbinary bch decoding (abstr.),'' {\em IEEE Transactions on
  Information Theory}, vol.~14, pp.~242--242, March 1968.

\bibitem{1054260}
J.~{Massey}, ``Shift-register synthesis and bch decoding,'' {\em IEEE
  Transactions on Information Theory}, vol.~15, pp.~122--127, January 1969.

\bibitem{rational1}
M.~Sudan, ``Notes on an efficient solution to the rational function
  interpolation problem,'' {\em Avaliable from
  \url{http://people.csail.mit.edu/madhu/FT01/notes/rational.ps}}, 1999.

\bibitem{rational2}
M.~Rosenblum, ``A fast algorithm for rational function approximations,'' {\em
  Avaliable from
  \url{http://people.csail.mit.edu/madhu/FT01/notes/rosenblum.ps}}, 1999.

\bibitem{pan2001matrix}
V.~Y. Pan, ``Matrix structures of vandermonde and cauchy types and polynomial
  and rational computations,'' in {\em Structured Matrices and Polynomials},
  pp.~73--116, Springer, 2001.

\bibitem{WentaosThesis}
W.~Huang, {\em Coding for Security and Reliability in Distributed Systems}.
\newblock PhD thesis, California Institute of Technology, 2017.

\end{thebibliography}
\bibliographystyle{ieeetr}

\section*{Supplementary Material}
\renewcommand{\thesubsection}{\Alph{subsection}}

\subsection{Algorithmic Illustration of LCC } \label{app:alg}

\begin{algorithm}
\caption{LCC Encoding (Precomputation)}\label{euclid}
\begin{algorithmic}[1]
\Procedure{Encode}{$X_1,X_2,...,X_K,T$}\Comment{Encode inputs variables according to LCC}
\State \textbf{generate} uniform random variables $Z_{K+1}, ..., Z_{K+T}$
\State \textbf{jointly compute} 
 $ \tilde{X}_i\gets \sum_{j\in[K]}X_j\cdot \prod_{k\in [K+T]\setminus\{j\}}\frac{\alpha_i-\beta_k}{\beta_j-\beta_k}+
    \sum_{j=K+1}^{K+T} Z_j\cdot \prod_{k\in [K+T]\setminus\{j\}}\frac{\alpha_i-\beta_k}{\beta_j-\beta_k}$
 for  $i=1,2,...,N$ using fast polynomial interpolation    
\State \textbf{return} $\tilde{X}_1,...,\tilde{X}_N$ \Comment{The coded variable assigned to worker $i$ is $\tilde{X}_i$}
\EndProcedure
\end{algorithmic}
\end{algorithm}

\begin{algorithm}
\caption{Computation Stage}
\begin{algorithmic}[1]
\Procedure{WorkerComputation}{$\tilde{X}$}\Comment{Each worker $i$ takes $\tilde{X}_i$ as input}
\State \textbf{return} $f(\tilde{X})$\Comment{Compute as if no coding is taking place }
\EndProcedure
\end{algorithmic}
\begin{algorithmic}[1]
\Procedure{Decode}{$S,A$}\Comment{Executed by master}
\State \textbf{wait} for a subset of fastest $N-S$ workers
\State $\mathcal{N} \gets$ identities of the fastest workers
\State $\{\tilde{Y}_i\}_{i\in\mathcal{N}} \gets$ results from the fastest workers
\State \textbf{recover} $Y_1,...,Y_K$ from $\{\tilde{Y}_i\}_{i\in\mathcal{N}}$ using fast interpolation or Reed-Solomon decoding \Comment{See Appendix \ref{app:comp} }
\State \textbf{return} $Y_1,...,Y_K$
\EndProcedure
\end{algorithmic}
\end{algorithm}

$\beta_1,\ldots,\beta_{K+T}$ and $\alpha_1,...,\alpha_N$ are global constants in $\bF$, satisfying\footnote{A variation of LCC is presented in Appendix \ref{app:ulcc}, by selecting different values of $\alpha_i$'s.} 
\begin{enumerate}
    \item %$\beta_i\neq \beta_j$ for $i\neq j$,
    $\beta_i$'s are distinct,
    \item %$\alpha_i\neq \alpha_j$ for $i\neq j$,
    $\alpha_i$'s are distinct,
    \item $\{\alpha_i\}_{i\in[N]}\cap\{\beta_j\}_{j\in[K]}=\varnothing$ (this requirement is alleviated if~$T=0$).
\end{enumerate}

% \begin{remark}\label{remark:systematicEnc}
% Note that by choosing~$\{\beta_i\}_{i=1}^K= \{\alpha_i\}_{i=1}^K$, the first $K$ workers exactly compute the~$K$ required results respectively. This provides a systematic coding design in the sense that the first $K$ workers are the \textit{systematic nodes} and the rest of the $N-K$ workers are \textit{parity nodes}.
% \end{remark}

\subsection{Coding Complexities of LCC} \label{app:comp}

By exploiting the algebraic structure of LCC, we can find efficient encoding and decoding algorithms with almost linear computational complexities. The encoding of LCC can be viewed as interpolating degree $K+T-1$ polynomials, and then evaluating them at $N$ points. It is known that both operations only require almost linear complexities: interpolating a polynomial of degree $k$ has a complexity of $O(k \log^2 k \log \log k)$, and evaluating it at any $k$ points requires the same \cite{kedlaya2011fast}. Hence, the total encoding complexity of LCC is at most $O(N \log^2 (K+T) \log \log (K+T) \dim \bV)$, which is almost linear to the output size of the encoder $O(N \dim \bV)$.

Similarly, when no security requirement is imposed on the system (i.e., $A=0$), the decoding of LCC can also be completed using polynomial interpolation and evaluation. An almost linear complexity $O(R \log^2 R \log \log R \dim \bU)$ can be achieved, where $R$ denotes the recovery threshold.    

A less trivial case is to consider the decoding algorithm when $A>0$, where the goal is essentially to interpolate a polynomial with at most $A$ erroneous input evaluations, or decoding a Reed-Solomon code. An almost linear time complexity can be achieved using additional techniques developed in \cite{Berlekamp:2006:NBD:2263266.2267601, 1054260,rational1,rational2}. Specifically, the following $2A-1$ \emph{syndrome variables} can be computed with a complexity of $O((N-S)\log^2 (N-S) \log\log (N-S) \dim \bU)$ using fast algorithms for polynomial evaluation and for transposed-Vandermonde-matrix multiplication \cite{pan2001matrix}. % which are used in the Berlekamp–Massey algorithm to locate errors. 
\begin{align}
    S_k\triangleq \sum_{i\in{\mathcal{N}}} \frac{\tilde{Y}_i \alpha_i^k}{\prod_{j\in\mathcal{N}\backslash\{i\} }(\alpha_{i}-\alpha_j)} && \forall k\in \{0,1,...,2A-1\}. 
\end{align}
According to \cite{Berlekamp:2006:NBD:2263266.2267601, 1054260}, the location of the errors (i.e., the identities of adversaries in LCC decoding) can be determined given these syndrome variables by computing its rational function approximation. Almost linear time algorithms for this operation are provided in \cite{rational1,rational2}, which only requires a complexity of $O(A\log^2 A \log\log A \dim \bU)$. After identifying the adversaries, the final results can be computed similar to the $A=0$ case. This approach achieves a total decoding complexity of $O((N-S)\log^2 (N-S) \log\log (N-S) \dim \bU)$, which is almost linear with respect to the input size of the decoder $O((N-S)\dim \bU)$. 
 
 %ince most related algorithms require quadratic complexity (e.g., \cite{Berlekamp:2006:NBD:2263266.2267601, 1054260}), it may seem that that the same level of complexity is also required for decoding LCC. However, an almost linear decoding complexity can also be achieved for the important case where the decoding could be a bottleneck (i.e., the output size $\dim \bU$ from each worker is large).
%Finally, we note that a set of syndrome variables, defined as follows and used in the Berlekamp-Massey algorithm, can be easily computed and used to detect errors. Specifically, given an arbitrary set $\mathcal{N}\subseteq [N]$ and list of scalar variables $\{y_i\}_{y\in\mathcal{N}}$, let
Finally, note that the adversaries can only affect a \emph{fixed} subset of $A$ workers' results for all entries. This decoding time can be further reduced by computing the final outputs entry-wise: for each iteration, ignore computing results from adversaries identified in earlier steps, and proceed decoding with the rest of the results.  

%Specifically, we exploit the fact that the adversaries can only affect a \emph{fixed} subset of $A$ results for all entries. Hence, by decoding the final outputs entry-wise, errors located from any iteration can be reused in all following steps, which reduces the computation cost.    

%\begin{algorithm}
%\caption{LCC Decoding for $A>0$}
%\begin{algorithmic}[1]
%\Procedure{Encode}{$X_1,X_2,...,X_K,T$}\Comment{Encode inputs variables according to LCC}
%\State \textbf{generate} uniform random variables $Z_{K+1}, ..., Z_{K+T}$
%\State \textbf{jointly compute} 
% $ \tilde{X}_i\gets \sum_{j\in[K]}X_j\cdot \prod_{k\in %[K+T]\setminus\{j\}}\frac{\alpha_i-\beta_k}{\beta_j-\beta_k}+
%    \sum_{j=K+1}^{K+T} Z_j\cdot \prod_{k\in [K+T]\setminus\{j\}}\frac{\alpha_i-\beta_k}{\beta_j-\beta_k}$
% for  $i=1,2,...,N$ using fast polynomial interpolation    
%\State \textbf{return} $\tilde{X}_1,...,\tilde{X}_N$ \Comment{The coded variable assigned to worker $i$ is %$\tilde{X}_i$}
%\EndProcedure
%\end{algorithmic}
%\end{algorithm}

%finally , note that although the algorithms proposed in literature to achieved the above complexityes,  use improved algorithms

\subsection{The MDS property of $U^{bottom}$}

\begin{lemma}\label{lemma:Uproperties}
The matrix~$U^{bottom}$ is an MDS matrix.
\end{lemma}

\begin{proof}
    First, let~$V\in \bF^{T\times N}$ be 
    \begin{align*}
        V_{i,j}&=\prod_{\ell \in [T]\setminus\{i\}}\frac{\alpha_j-\beta_{\ell+K}}{\beta_{i+K}-\beta_{\ell+K}}.
    \end{align*}
    It follows from the resiliency property of LCC that by having~$(\tilde{X}_1,\ldots,\tilde{X}_N)=(X_1,\ldots,X_T)\cdot V$, the master can obtain the values of~$X_1,\ldots,X_T$ from any~$T$ of the~$\tilde{X}_i$'s. This is one of the alternative definitions for an MDS code, and hence, $V$ is an MDS matrix.
    
    To show that~$U^{bottom}$ is an MDS matrix, 
    % we show that for every set~$\cT$ of~$T$ workers, the submatrix~$U_\cT^{bottom}$ of~$U^{bottom}$ is invertible.
    % %Let~$U'_L\in\bF^{T\times N}$ be the submatrix of~$U'$ which consists of the~$T$ bottom rows. Clearly, it suffices to prove that~$U'_L$ is an MDS matrix. 
    % To this end, let~$V\in\bF^{T\times N}$ be the encoding matrix which results from applying~\eqref{equation:encoding} with the elements~$\{\alpha_i\}_{i\in[N]}$ and~$\{\beta_i\}_{i=K+1}^{K+T}$, and notice that~$V$ is an MDS matrix by Proposition~\ref{proposition:MDS}.
    it is shown that~$U^{bottom}$ can be obtained from~$V$ by multiplying rows and columns by nonzero scalars. Let~$[T:K]\triangleq\{ T+1,T+2,\ldots,T+K \}$, and notice that for~$(s,r)\in[T]\times [N]$, entry~$(s,r)$ of~$U^{bottom}$ can be written as
    \begin{align*}
        \prod _{t\in [K+T]\setminus \{ s+K \} } \frac{\alpha_r-\beta_t}{\beta_{s+K}-\beta_t}&=\prod_{t\in[K]}\frac{\alpha_r-\beta_t}{\beta_{s+K}-\beta_t}\cdot\\
        &~\prod_{t\in [K:T]\setminus \{s+K\}}\frac{\alpha_r-\beta_t}{\beta_{s+K}-\beta_t}.
    \end{align*}
    Hence,~$U^{bottom}$ can be written as
    \begin{align}\label{equation:Vmatrix}
        U^{bottom}=&\diag\left( \left( \prod_{t\in[K]}\frac{1}{\beta_{s+K}-\beta_t} \right)_{s\in[T]} \right) \cdot V \cdot\nonumber\\
        &\diag \left( \left( \prod_{t\in[K]}(\alpha_r-\beta_t) \right)_{r\in [N]}\right),
    \end{align}
    where~$V$ is a~$T\times N$ matrix such that
    \begin{align*}
        V_{i,j}=\prod_{t\in [T]\setminus\{i\}}\frac{\alpha_{j}-\beta_{t+K}}{\beta_{i+K}-\beta_{t+K}}.
    \end{align*}
    Since~$\{ \beta_t \}_{t=1}^K\cap \{ \alpha_r \}_{r=1}^N=\varnothing$, and since all the~$\beta_i$'s are distinct, it follows from~\eqref{equation:Vmatrix} that~$U^{bottom}$ can be obtained from~$V$ by multiplying each row and each column by a nonzero element, and hence~$U^{bottom}$ is an MDS matrix as well.
\end{proof}

\subsection{The Uncoded Version of LCC} \label{app:ulcc}

In Section \ref{section:GeneralDescription}, we have described the LCC scheme, which provides an~$S$-resilient, $A$-secure, and~$T$-private scheme as long as   $(K+T-1)\deg f+S+2A+1\leq N$. 
Instead of explicitly following the same construction, a variation of LCC can be made by instead selecting the values of $\alpha_i$'s from the set $\{\beta_j\}_{j\in[K]}$ (not necessarily distinctly). 

We refer to this approach as the \emph{uncoded version of LCC}, which essentially recovers the \emph{uncoded repetition} scheme, which simply replicates each $X_i$ onto multiple workers. By replicating every~$X_i$ between~$\lfloor N/K \rfloor$ and~$\lceil N/K \rceil$ times, it can tolerate at most $S$ stragglers and $A$ adversaries, whenever
\begin{align}
    S+2A\leq\lfloor N/K \rfloor -1, 
\end{align}
which achieves the optimum resiliency and security when the number of workers is small and no data privacy is required (specifically, $N< K\deg f-1$ and $T=0$, see Section \ref{sec:converses}). % However, uncoded repetition does not support the privacy requirement.

%todo: see sec converse

%{Instead of storing the input data using LCC, another possible approach is to use \emph{uncoded repetition} -- that is, by replicating every~$X_i$ between~$\lfloor N/K \rfloor$ and~$\lceil N/K \rceil$ times, when there is no privacy requirement (i.e., $T=0$). This uncoded approach achieves a resiliency-security pair $(S,A)$ whenever $N\geq K(S+2A+1)$, which could improve LCC when the number of workers is small (specifically, $N< K\deg f-1$). However, uncoded repetition does not support the privacy requirement.}

{When privacy is taken into account (i.e., $T>0$), an alternative approach in place of repetition is to instead store each input variable using Shamir's secret sharing scheme \cite{Shamir:1979:SS:359168.359176} over $\lfloor N/K \rfloor$ to~$\lceil N/K \rceil$ machines. This approach achieves any $(S,A,T)$ tuple whenever $N\geq K(S+2A+\textup{deg}\ f  \cdot T + 1)$. However, it does not improve LCC.}

\subsection{Proof of Lemma \ref{lemma:rec}} \label{pl:rec}

We start by defining the following notations. For any multilinear function $f$ defined on $\bV$ with degree $d$, let
$X_{i,1},X_{i,2},...,X_{i,d}$ denote its $d$ input entries (i.e., $X_i= (X_{i,1},X_{i,2},...,X_{i,d})$ and $f$ is linear with respect to each entry). Let $\bV_1,...,\bV_d$ be the vector space that contains the values of the entries. For brevity, we denote $\deg f$ by $d$ in this appendix. We first provide the  proof of inequality (\ref{ineq:rec}).

\begin{proof}[Proof of inequality (\ref{ineq:rec})]
Without loss of generality, we assume both the encoding and decoding functions are deterministic in this proof, as the randomness does not help with decodability.\footnote{Note that this argument requires the assumption that the decoder does not have access to the random keys, as assumed in Section \ref{section:formulation}.} %If this assumption can be violated, improved designs could be found, see Appendix \ref{app:discussion}. }
Similar to \cite{yu2018straggler}, we define the minimum recovery threshold, denoted by $R^*(N,K,f)$, as the minimum number of workers that the master has to wait to guarantee decodability, among all linear encoding schemes. Then we essentially need to prove that $R^*(N,K,f)\geq R^*_{\textup{LCC}}(N,K,f)$, i.e., $R^*(N,K,f)\geq (K-1)d+1$ when $N\ge Kd-1$, and $R^*(N,K,f)\geq N-\lfloor N/K \rfloor +1$ when $N< Kd-1$.

Obviously $R^*(N,K,f)$ is a non-decreasing function with respect to $N$. Hence, it suffices to prove that  $R^*(N,K,f)\geq N-\lfloor N/K \rfloor +1$ when $N\leq Kd-1$.  We prove this converse bound by induction.

(a) If $d=1$, then $f$ is a linear function, and we aim to prove $R^*(N,K,f)\geq N+1$ for $N\leq K-1$. This essentially means that no valid computing schemes can be found when $N<K$.
Assuming the opposite, suppose we can find a valid computation design using at most $K-1$ workers, then there is a decoding function that computes all $f(X_i)$'s given the results from these workers. 

Because the encoding functions are linear, 
%Under this condition, 
we can thus find a non-zero vector $(a_1,...,a_K)\in \bF^{K}$ such that when $X_i=a_iV$ for any $V\in \bV$, the coded variable $\tilde{X}_i$ stored by any worker equals the padded random key, which is a constant. This leads to a fixed output from the decoder. 
On the other hand, because $f$ is assumed to be non-zero, the computing results $\{f(X_i)\}_{i\in[K]}$ is variable for different values of $V$, which leads to a contradiction. Hence, we have prove the converse bound for $d=1$. %\Netanel{$K_f^*\rightarrow K^*$?}

%On the other hand, the decoder can only return a constant given this class of inputs   

(b) Suppose we have a matching converse for any multilinear function with $d=d_0$. We now prove the lower bound for any  multilinear function $f$ of degree $d_0+1$. Similar to part (a), it is easy to prove that $R^*(N,K,f)\geq N+1$ for $N\leq K-1$. Hence, we focus on $N\geq K$.
%To prove this result we only need to consider two scenarios:
%(1) If $K^*_f=N$, the converse trivially holds. 
%(2) If So we focus on the scenarios where a computing design can be found that tolerates at least $1$ stragglers (i.e., $K^*_f<N$).  

The proof idea is to construct a multilinear function $f'$ with degree $d_0$ based on function $f$, and to lower bound the minimum recovery threshold of $f$ using that of $f'$. More specifically, this is done by showing that given any computation design for function $f$, a computation design can also be developed for the corresponding $f'$, which achieves a recovery threshold that is related to that of the scheme for $f$.

In particular, for any non-zero function $f(X_{i,1},X_{i,2},...,X_{i,d_0+1})$, we let $f'$ be a function which takes inputs $X_{i,1},X_{i,2},...,X_{i,d_0}$ and returns a linear map, such that given any  $X_{i,1},X_{i,2},...,X_{i,d_0+1}$, we have $f'(X_{i,1},X_{i,2},...,X_{i,d_0})(X_{i,d_0+1})=f(X_{i,1},X_{i,2},...,X_{i,d_0+1})$. One can verify that $f'$ is a multilinear function with degree $d_0$,
%In particular, for any non-zero function $f(X_{i,1},X_{i,2},...,X_{i,d_0+1})$, we can find $V\in\bV_{d_0+1}$, such that $f(X_{i,1},X_{i,2},...,X_{i,d_0},V)$ as a function of $(X_{i,1},X_{i,2},...,X_{i,d_0})$ is non-zero. We define $f'(X_{i,1},X_{i,2},...,X_{i,d_0})=f(X_{i,1},X_{i,2},...,X_{i,d_0},V)$, which is a multilinear function with degree $d_0$.
Given parameters $K$ and $N$, we now develop a computation strategy for $f'$ for a dataset of $K$ inputs and a cluster of $N'\triangleq N-K$ workers, which achieves a recovery threshold of $R^*(N,K,f)-(K-1)$. % unless $R^*(N,K,f)\geq N$.
We construct this computation strategy based on an encoding strategy of $f$ that achieves the recovery threshold $R^*(N,K,f)$. For brevity, we refer to these two schemes as the $f'$-scheme and $f$-scheme respectively.

Because the encoding functions are linear, we consider the encoding matrix, denoted by $G\in \bF^{K\times N}$, and defined as the coefficients of the encoding functions $\tilde{X}_i=\sum_{j=1}^{K} X_jG_{ji}+\tilde{z}_i$, where $\tilde{z}_i$ denotes the value of the random key padded to variable $\tilde{X}_i$. Following the same arguments we used in the $d=1$ case, the left null space of $G$ must be $\{0\}$. Consequently, the rank of $G$ equals $K$, and we can find a subset $\mathcal{K}$ of $K$ workers such that the corresponding columns of $G$ form a basis of $\bF^{K}$. %We construct the computation strategy for $f'$ by letting each of the $N-K$ workers store the coded version of $(X_{i,1},X_{i,2},...,X_{i,d_0})$ stored by a worker in $[N]\backslash \mathcal{K}$, respectively.
Hence, we can construct the $f'$-scheme by letting each of the~$N'\triangleq N-K$ workers store the coded version of~$(X_{i,1},X_{i,2},\ldots,X_{i,d_0})$ that is stored by a unique respective worker in $[N]\setminus\mathcal{K}$ in $f$-scheme.\footnote{For breivity, in this proof we instead index these $N-K$ workers also using the set $[N]\setminus\mathcal{K}$, following the natural bijection.}

%We prove Lemma \ref{lemma:rec} by lower bounding the recovery threshold for two cases ($T=0$ and $T>0$) separately.

Now it suffices to prove that the above construction achieves a recovery threshold of $R^*(N,K,f)-(K-1)$. Equivalently, we need to prove that given any subset $\mathcal{S}$ of $[N]\backslash \mathcal{K}$ of size $R^*(N,K,f)-(K-1)$, the values of $f(X_{i,1},X_{i,2},...,X_{i,d_0},x)$ for any $i\in[K]$ and $x\in\bV$ are decodable from the computing results of workers in $\mathcal{S}$. 

We exploit the decodability of the computation design for function $f$. For any $j\in\mathcal{K}$, the set $\mathcal{S}\cup \mathcal{K}\backslash\{j\}$ has size $R^*(N,K,f)$. Consequently, 
for any vector $(x_{1,d_0+1},...,x_{K,d_0+1})\in\bV_{d_0+1}^K$, we have that $\{ f(X_{i,1},X_{i,2},...,X_{i,d_0},x_{i,d_0+1})\}_{i\in[K]}$ is decodable given the results from workers in $\mathcal{S}\cup \mathcal{K}\backslash\{j\}$ computed in $f$-scheme, if each $x_{i,d_0+1}$ is used as the $(d_0+1)$th entree for each input.

Because columns of $G$ with indices in $\mathcal{K}$ form a basis of $\bF^{K}$, we can find values for each input $X_{i,d_0+1}$ such that workers in $\mathcal{K}$ would store $0$ for the $X_{i,d_0+1}$ entry in the $f$-scheme. We denote these values by $\bar{x}_{1,d_0+1},...,\bar{x}_{K,d_0+1}$. Note that if these values are taken as inputs, workers in $\mathcal{K}$ would return constant $0$ {due to the multilinearity of $f$.} Hence, decoding $f(X_{i,1},X_{i,2},...,X_{i,d_0},\bar{x}_{i,d_0+1})$ only requires results from workers not in $\mathcal{K}$, i.e., it can be decoded given computing results from workers in $\mathcal{S}$ using the $f$-scheme. Note that these results can be directly computed from corresponding results in the $f'$-scheme. We have proved the decodability of $f(X_{i,1},X_{i,2},...,X_{i,d_0},x)$ for $x=\bar{x}_{i,d_0+1}$. % is also decodable given computing results from workers in $\mathcal{S}$ using the $f'$-scheme.    

Now it remains to prove the decodability of $f(X_{i,1},X_{i,2},...,X_{i,d_0},{x})$ for each $i$ for general ${x}\in\bV$. %To that end, we define the following notations. 
For any $j\in\mathcal{K}$, let $\boldsymbol{a}^{(j)}\in\bF^K$ be a non-zero vector that is orthogonal to all columns of $G$ {with indices in  $\mathcal{K}\backslash\{j\}$}. %Fixing $i$, let $x'\triangleq x-\bar{x}_{i,d_0+1}$.
If $a_{i}^{(j)} x+\bar{x}_{i,d_0+1}$ is used for each input $X_{i,d_0+1}$ in the $f$-scheme, then
%except for column $j$, 
workers in $\mathcal{K}\backslash\{j\}$ would store $0$ for the $X_{i,d_0+1}$ entry, and return constant $0$ {due to the multilinearity of $f$.}~ Recall that $f(X_{i,1},X_{i,2},...,X_{i,d_0},a_{i}^{(j)} x+\bar{x}_{i,d_0+1})$ is assumed to be decodable in the $f$-scheme given results from workers in $\mathcal{S}\cup \mathcal{K}\backslash\{j\}$ . Following the same arguments above, one can prove that $f(X_{i,1},X_{i,2},...,X_{i,d_0},a_{i}^{(j)} x+\bar{x}_{i,d_0+1})$ is also decodable using the $f'$-scheme.
Hence, the same applies for $a_{i}^{(j)} f(X_{i,1},X_{i,2},...,X_{i,d_0},x)$ due to multilinearity of $f$.

Because columns of $G$ with indices in $\mathcal{K}$ form a basis of $\bF^{K}$, the vectors $\boldsymbol{a}^{(j)}$ for $j\in\mathcal{K}$ also from a basis. Consequently, for any $i$ there is a non-zero $a_{i}^{(j)}$, and thus $f(X_{i,1},X_{i,2},...,X_{i,d_0},x)$ is decodable. This completes the proof of decodability.
%we can find  $f'(X_{i,1},X_{i,2},...,X_{i,d_0})$, which equals $f(X_{i,1},X_{i,2},...,X_{i,d_0},V)$, is also decodable given results from workers in $\mathcal{S}$ for any $i\in[K]$. On the other hand, note that the computing results for each worker in $\mathcal{S}$ given each $\boldsymbol{a}^{(j)}$ can also be computed using the results from the same workers when computing $f'$. Hence, the decoder for function $f'$ can first recover the computing results for workers in $\mathcal{S}$ for function $f$, and then proceed to decoding the final result. 
%s we have completed the proof of decodability.

To summarize, we have essentially proved that $R^*(N,K,f)-(K-1)\geq R^*(N-K,K,f')$. %whenever a valid scheme compting $f$ can be found given $N$ and $K$ (i.e., $R^*(N,K,f)\leq N$). % $N\geq 2K$, and $K^*_f(K,N)-(K-1)> N-K$ otherwise.
We can verify that the converse bound $R^*(N,K,f)\geq N-\lfloor N/K \rfloor +1$ under the condition $N\leq Kd-1$ can be derived given the above result and the induction assumption, for any function $f$ with degree $d_0+1$.

(c) Thus, a matching converse holds for any $d\in\mathbb{N}_+$, which proves inequality (\ref{ineq:rec}).
\end{proof}

Now we proceed to prove the rest of Lemma \ref{lemma:rec}, explicitly, we aim to prove that the recovery threshold of any $T$-private encoding scheme is at least $R_\textup{LCC}(N,K,f)+T\cdot \deg f$. Inequality (\ref{ineq:rec}) essentially covers the case for $T=0$. Hence, we focus on $T>0$. 
To simplify the proof, we prove a stronger version of this statement: when $T>0$, any valid $T$-private encoding scheme uses at least $N\geq R_\textup{LCC}(N,K,f)+T\cdot \deg f$ workers. Equivalently, we aim to show that $N\geq (K+T-1) \deg f+1$ for any such scheme.

We prove this fact using an inductive approach. To enable an inductive structure, we prove a even stronger converse by considering a more general class of computing tasks and a larger class of encoding schemes, formally stated in the following lemma.

%\begin{lemma}\label{lemma:strong}
%Consider a dataset $X_1,...,X_K\in\bF^{d}$, a set of $N$ workers where each can take a coded variable in $\bF^{d}$ and return the product of its elements, and a computing task where the master aim to recover $Y_i\triangleq X_{i,1}\cdot...\cdot X_{i,d-1} \cdot G_{i,d}$, where each $G_{i,d}$ is a non-zero fixed linear combination of the $d$th entries of the dataset. If the $i$th entry of each coded variable is encoded using a $T_i$-private linear scheme and $T_i>0$ for $i=1,...,d-1$, then any valid valid computing scheme requires $N\geq \sum_{i=1}^{d} T_i +(K-1)(d-1)+r$, where $r$ is the rank of variables $\{G_{i,d}\}_{i\in[K]}$. \footnote{This lemma  only applies when the computation task is non-trivial, i.e., $r\neq 0$. Otherwise, no worker is needed. } 
%\end{lemma}

\begin{lemma}\label{lemma:strong}
Consider a dataset with inputs $X\triangleq (X_1,...,X_K)\in(\bF^{d})^K$, and an input vector $\Gamma \triangleq(\Gamma_1,...,\Gamma_K)$ which belongs to a given subspace of $\bF^{K}$ with dimension $r>0$; a set of $N$ workers where each can take a coded variable in $\bF^{d+1}$ and return the product of its elements; and a computing task where the master aim to recover $Y_i\triangleq X_{i,1}\cdot...\cdot X_{i,d} \cdot \Gamma_{i}$. If the inputs entries are encoded separately such that each of the first $d$ entries assigned to each worker are some $T_{\textup{X}}>0$-privately linearly coded version of the corresponding entries of $X_i$'s,
 and  the $(d+1)$th entry assigned to each worker is a $T$-privately\footnote{For this lemma, we assume that no padded random variable is used for a $0$-private encoding scheme.} linearly coded version of $\Gamma$, moreover, if each $\Gamma_i$ (as a variable) is non-zero, then any valid computing scheme requires $N\geq (T_\textup{X}+K-1)d+T+r$. %\footnote{This lemma  only applies when the computation task is non-trivial, i.e., $r\neq 0$. Otherwise, no worker is needed. } 
\end{lemma}

%todo need uniform for eliminating the degrees separately

\begin{proof}

Lemma \ref{lemma:strong} is proved by induction with respect to the tuple $(d,T,r)$. Specifically, we prove that (a) Lemma \ref{lemma:strong} holds when $(d,T,r)=(0,0,1)$; (b) If Lemma \ref{lemma:strong} holds for any $(d,T,r)=(d_0,0,r_0)$, then it holds when $(d,T,r)=(d_0,0,r_0+1)$; (c) If Lemma \ref{lemma:strong} holds for any $(d,T,r)=(d_0,0,r_0)$, then it holds when $(d,T,r)=(d_0,T,r_0)$ for any $T$; (d) If Lemma \ref{lemma:strong} holds for any $d=d_0$ and arbitrary values of $T$ and $r$, then it holds if $(d,T,r)=(d_0+1,0,1)$. Assuming the correctness of these statements, Lemma \ref{lemma:strong} directly follows by induction's principle. Now we provide the proof of these statements as follows.

(a).  When $(d,T,r)=(0,0,1)$, we need to show that at least $1$ worker is needed. This directly follows from the decodability requirement, because the master aims to recover a variable, and at least one variable is needed to provide the information. 

(b). Assuming that for any $(d,T,r)=(d_0,0,r_0)$ and any $K$ and $T_\textup{X}$, any valid computing scheme requires $N\geq  (T_\textup{X}+K-1)d_0+r$ workers, we need to prove that for $(d,T,r)=(d_0,0,r_0+1)$, at least $(T_\textup{X}+K-1)d_0+r_0+1$ workers are needed. We prove this fact by fixing an arbitrary valid computing scheme for $(d,T,r)=(d_0,0,r_0+1)$. For brevity, let $\tilde{\Gamma}_i$ denotes the coded version of $\Gamma$ stored at worker $i$. We consider the following two possible scenarios:  (i) there is a worker $i$ such that $\tilde{\Gamma}_i$ is not identical (up to a constant factor) to any variable $\Gamma_j$, or (ii) for any worker $i$, $\tilde{\Gamma}_i$ is identical (up to a constant factor) to some $\Gamma_j$.

For case (i), similar to the ideas we used to prove inequality (\ref{ineq:rec}), it suffices to show that if the given computing scheme uses $N$ workers, we can construct another computation scheme achieving the same $T_X$, for a different computing task with parameters $d=d_0$ and $r= r_0$, using at most $N-1$ workers.

Recall that  we assumed that there is a worker $i$, such that $\tilde{\Gamma}_i$ is not identical (up to a constant factor) to any $\Gamma_j$. %If  $\tilde{\Gamma}_i$ is not a constant, 
We can always restrict the value of $\Gamma$ to a subspace with dimension $r_0$, such that $\tilde{\Gamma}_i$ becomes a constant $0$. After this operation, from the computation results of the rest $N-1$ workers, the master can recover a computing function with $r=r_0$ and non-zero $\Gamma_j$'s, which provides the needed computing scheme. 

For case (ii), because each $\Gamma_j$ is assumed to be non-zero, we can partition the set of indices $j$ into distinct subsets, such that any $j$ and $j'$ are in the same subset iff $\Gamma_j$ is a constant multiple of $\Gamma_{j'}$. We denote these subsets by $\mathcal{J}_1,...,\mathcal{J}_m$. Moreover, for any $k\in[m]$, let $\mathcal{I}_k$ denote the subset of indices $i$ such that $\tilde{\Gamma}_i$ is identical (up to a constant factor) to $\Gamma_j$ for $j$ in $\mathcal{J}_k$. 

Now for any $k\in[m]$, we can restrict the value of $\Gamma$ to a subspace with dimension $r_0$, such that $\Gamma_j$ is zero for any $j\in\mathcal{J}_k$. After applying this operation, from the computation results of workers in $[N]\backslash \mathcal{I}_k$, the master can recover a computing function with $r=r_0$, where $K'=K-|\mathcal{J}_k|$ sub-functions has non-zero $\Gamma_j$'s. By applying the induction assumption on this provided computing scheme, we have $N-|\mathcal{I}_k|\geq (T_\textup{X}+K-|\mathcal{J}_k|-1)d_0+r_0$.
By taking the summation of the this inequality over $k\in [m]$, we have 
\begin{align}\label{ineq:strong}
    Nm-\sum_{k=1}^m|\mathcal{I}_k|\geq (T_\textup{X}m+Km-K-m)d_0+r_0m. 
\end{align}
Recall that for any worker $i$, $\tilde{\Gamma}_i$ is identical (up to a constant factor) to some $\Gamma_j$, we have $\cup_{k\in[m]} \mathcal{I}_k=[N]$. Thus, $\sum_{k}|\mathcal{I}_k|\geq N$. Consequently, inequality (\ref{ineq:strong}) implies that 
\begin{align}\label{ineq:strong2}
    Nm-N\geq (T_\textup{X}m+Km-K-m)d_0+r_0m. 
\end{align}
Note that $r_0+1>1$, which implies that at least two ${\Gamma}_j$'s are not identical up to a constant factor. Hence, $m-1>0$, and (\ref{ineq:strong2}) is equivalently
\begin{align}
    N&\geq \frac{(T_\textup{X}m+Km-K-m)d_0+r_0m}{m-1}\\
    &= (T_\textup{X}+K-1)d_0+r_0+\left((T_\textup{X}-1)d_0+r_0\right)\frac{1}{m-1}.
\end{align}
Since $T_\textup{X}$ and $r_0$ are both positive, we have $(T_\textup{X}-1)d_0+r_0>0$. Consequently, $ \left((T_\textup{X}-1)d_0+r_0\right)\frac{1}{m-1}>0$, and we have 
%It is then easy to show that there is an $k_0\in[m]$, such that  $|\mathcal{I}_{k_0}|\geq $
\begin{align}
    N\geq(T_\textup{X}+K-1)d_0+r_0+1,
\end{align}
which proves the induction statement.

%we can find a variable $j_0$, such that at least $\lceil \frac{N}{K} \rceil$ workers stores a variable that is a constant multiple of $\Gamma_{j_0}$ 

%the same task using $T$ less workers, while providing the same privacy requirement for entries $X_{i,1},...,X_{i,d-1}$.  

%first select t to be 0

(c). Assuming that for any $(d,T,r)=(d_0,0,r_0)$, any valid computing scheme requires $N\geq (T_\textup{X}+K-1)d_0+r_0$ workers, we need to prove that for $(d,T,r)=(d_0,T_0,r_0)$, $N\geq (T_\textup{X}+K-1)d_0+T_0+r_0$. Equivalently, we aim to show that for any $T_0>0$, in order to provide $T_0$-privacy to the $d_0+1$th entry, $T_0$ extra worker is needed.
Similar to the earlier steps, we consider an arbitrary valid computing scheme for $(d,T,r)=(d_0,T_0,r_0)$ that uses $N$ workers. We aim to construct a new scheme for $(d,T,r)=(d_0,0,r_0)$, for the same computation task and the same $T_\textup{X}$, which uses  at most $N-T_0$ workers.

Recall that if an encoding scheme is $T_0$ private, then given any subset of at most $T_0$ workers, denoted by $\mathcal{T}$, we have $I(\Gamma;\tilde{\Gamma}_\mathcal{T})=0$. Consequently, conditioned on $\tilde{\Gamma}_\mathcal{T}=0$, the entropy of the variable $\Gamma$ remains unchanged. This indicates that $\Gamma$ can be any possible value when $\tilde{\Gamma}_\mathcal{T}=0$. Hence, we can let the values of the padded random variables be some linear combinations of the elements of $\Gamma$, such that worker in $\mathcal{T}$ returns constant $0$. 

Now we construct an encoding scheme as follows. Firstly it is easy to show that when the master aims to recover a non-constant function, at least $T_0+1$ workers are needed to provide non-zero information regarding the inputs. Hence, we can arbitrarily select a subset of $T_0$ workers, denoted by $\mathcal{T}$. 
As we have proved, we can find fix the values of the padded random variables such that $\tilde{\Gamma}_{\mathcal{T}}=0$. Due to multilinearity of the computing task, these workers in $\mathcal{T}$ also returns constant $0$. Conditioned on these values, the decoder essentially computes the final output only based on the rest $N-T_0$ workers, which provides the needed computing  scheme. Moreover, as we have proved that the values of the padded random variables can be chosen to be some linear combinations of the elements of $\Gamma$, our obtained computing scheme encodes $\Gamma$ linearly.
This completes the proof for the induction statement.

(d). Assuming that for any $d=d_0$ and arbitrary values of $T$ and $r$, any valid computing scheme requires $N\geq (T_\textup{X}+K-1)d_0+T+r$ workers, we need to prove that for $(d,T,r)=(d_0+1,0,1)$, $N\geq (T_\textup{X}+K-1)(d_0+1)+1$. Observing that for any computing task with $r=1$, by fixing an non-zero $\Gamma$, it essentially computes $K$ functions where each multiplies $d_0$ variables. Moreover, for each function, by viewing the first $(d_0-1)$ entries as a vector $X'_i$ and by viewing the last entry as a scalar $\Gamma'_i$, it essentially recovers the case where the parameter $d$ is reduced by $1$, $K$ remain unchanged, and $r$ equals $K$. By adapting any computing scheme in the same way, we have $T_\textup{X}$ remain unchanged, and $T$ becomes $T_\textup{X}$. Then by induction assumption, any computing scheme for  $(d,T,r)=(d_0+1,0,1)$ requires at least $(T_\textup{X}+K-1)d_0+T_\textup{X}+K=(T_\textup{X}+K-1)(d_0+1)+1$ workers.   
\end{proof}

\begin{remark}
Using exactly the same arguments, Lemma \ref{lemma:strong} can be extended to the case where the entries of $X$ are encoded under different privacy requirements. Specifically, if the $i$th entry is $T_i$-privately encoded, then at least $\sum_{i=1}^{d}T_i+(K-1)d+T+r$ worker is needed. Lemma \ref{lemma:strong} and this extended version are both tight, in the sense for any parameter values of $d$, $K$ and $r$, there are computing tasks where a computing scheme that uses the matching number of workers can be found, using  constructions similar to the Lagrange coded computing.  
\end{remark}

Now using Lemma \ref{lemma:strong}, we complete the proof of Lemma \ref{lemma:rec} for $T>0$. Similar to the proof ideas for inequality (\ref{ineq:rec}) part (a), we consider any multilinear function $f$ with degree $d$, and we find constant vectors $V_1,...,V_d$, such that $f(V_1,...,V_d)$ is non-zero. Then by restricting the input variables to be constant multiples of $V_1,...,V_d$, this computing task reduces to multiplying $d$ scalars, given $K$ inputs. As stated in Lemma \ref{lemma:strong} and discussed in part (d) of its induction proof, such computation requires $(T+K-1)d+1$ workers. This completes the proof of Lemma \ref{lemma:rec}.

%TODO:merge

%Lemma \ref{lemma:t0} is proved in Appendices \ref{pl:t0}. The main proof idea is to show that for any computing strategy that tries to tolerate more stragglers than specified, there would be scenarios where all available computing results are degenerated {(i.e., constants)}, while the computing results needed by the master are variable, violating the decodability requirement.  

\subsection{Optimality on the Resiliency-Security-Privacy Tradeoff for Multilinear Functions } \label{pl:security}

%In this appendix, we prove that LCC provides security against the maximum possible number of adversaries. By comparing Lemma \ref{lemma:t0} and Theorem \ref{lemma:security}, we essentially need to show that the maximum possible number of adversaries that any linear encoding scheme can tolerate, is no greater than half of the corresponding number for stragglers.

In this appendix, we prove the first part of Theorem \ref{thm:opt} using Lemma \ref{lemma:rec}. Specifically, we aim to prove that LCC achieves the optimal trade-off between resiliency, security, and privacy for any multilinear function $f$. By comparing Lemma \ref{lemma:rec} and the achievability result presented in Theorem \ref{thm:lcc} and Appendix \ref{app:ulcc}, we essentially need to show that for any linear encoding scheme that can tolerates $A$ adversaries and $S$ stragglers, it can also tolerate $S+2A$ stragglers.

This converse can be proved by connecting the straggler mitigation problem and the adversary tolerance problem using the extended concept of Hamming distance for coded computing, which is defined in \cite{yu2018straggler}. 
Specifically, given any (possibly random) encoding scheme, its hamming distance is defined as the minimum integer, denoted by $d$, such that for any two instances of input $X$ whose outputs $Y$ are different, and for any two possible realizations of the $N$ encoding functions, the computing results given the encoded version of these two inputs, using the two lists of encoding functions respectively, differs for at least $d$ workers.

It was shown in \cite{yu2018straggler} that this hamming distance behaves similar to its classical counter part: an encoding scheme is $S$-resilient and $A$-secure whenever $S+2A\leq d-1$. 
\Added{Hence, for any encoding scheme that is $A$-secure and $S$-reselient, it has a hamming distance of at least $S+2A+1$. Consequently it can tolerate $S+2A$ stragglers. }
Combining the above and Lemma \ref{lemma:rec}, we have completed the proof.

\subsection{Optimality on the Resiliency-Privacy Tradeoff for General Multivariate Polynomials} \label{pl:resiliency}

In this appendix, we prove the second part of Theorem \ref{thm:opt} using Lemma \ref{lemma:rec}. Specifically, we aim to prove that LCC achieves the optimal trade-off between resiliency and privacy, for general multivariate polynomial $f$. 
The proof is carried out by showing that for any function $f$ that allows $S$-resilient $T$-private designs, there exists a multilinear function with the same degree for which a computation scheme can be found that achieves the same requirement.

%Then by noting that the achievable region specified in Definition \ref{def:r} only depends on $\deg f$ given any parameters $N$ and $K$, the optimality natrually follows.

%. For clarity, in this section (and relative appendices) we let $K^*_{f}(K,N)$ denote the minimum recovery threshold given $f$, $K$, and $N$.
%ower bounding the minimum recovery threshold 

%using the  can be used for computing a  for any  leads to 
%first prove the converse  

%To demonstrate the main proof ideas, we start by proving this converse for the special case of multi-linear functions. We summarize this partial result in the following lemma:

Specifically, given any function $f$ with degree $d$, we aim to provide an explicit construction of an multilinear function, denoted by $f'$, which achieves the same requirements. The construction satisfies certain properties to ensure this fact. Both the construction and the properties are formally stated in the following lemma (which is proved in Appendix \ref{pl:basic}):
 %a non-zero, multilinear function $f'$ with the same degree such that satisfies the above requirement. More specifically, given any computation scheme that tolerates $S$ stragglers for $f$, we show that  a computation design can be constructed for $f'$ that 
 
 %, and we aim to prove $K^*_{f}(K,N)\geq K^*_{f'}(K,N)$ by constructing a computation design of $f'$ that is based on a computation design of $f$ and achieves the same recovery threshold. The construction and its properties are stated in the following lemma (which is proved in Appendix \ref{pl:basic}):
\begin{lemma}\label{lemma:basic}
Given any function $f$ of degree $d$, let $f'$ be a map from $\bV^{d}\rightarrow \bU$ such that 
 $   f'(Z_{1},...,Z_{d})=\sum_{\mathcal{S}\subseteq[d]}{(-1)^{|\mathcal{S}|}f(\sum_{j\in\mathcal{S}}Z_{j})}$ for any $\{Z_{j}\}_{j\in[d]}\in\bV^{d}$. Then $f'$ is multilinear with respect to the $d$ inputs. Moreover, if the characteristic of the base field $\mathbb{F}$ is $0$ or greater than $d$, then $f'$ is non-zero.
\end{lemma}
%\vspace{-3mm}

Assuming the correctness of Lemma \ref{lemma:basic}, it %have constructed a multilinear polynomial with the same degree, we
suffices to prove that $f'$ enables computation designs that tolerates at least the same number of stragglers, and provides at least the same level of data privacy, compared to that of $f$. %Recall the notation $K^*_f(K,N)$ we defined earlier in this section, we aim to prove $K^*_f(K,N)\geq K^*_{f'}(K,N)$ for any $K$ and $N$.
We prove this fact by constructing such computing schemes for $f'$ given any design for $f$.

Note that $f'$ is defined as a linear combination of functions $f(\sum_{j\in\mathcal{S}}Z_{j})$, each of which is a composition of a linear map and $f$. Given the linearity of the encoding design, any computation scheme of $f$ can be directly applied to any of these functions, achieving the same resiliency and privacy requirements. Since the decoding functions are linear, the same scheme also applies to linear combinations of them, which includes $f'$. 
Hence, the resiliency-privacy tradeoff achievable for $f$ can also be achieved by $f'$. This concludes the proof. %Theorem \ref{lemma:resiliency} follows immediately.

\subsection{Proof of Lemma \ref{lemma:basic}} \label{pl:basic}

%\Netanel{I got confused here, since until now we had~$X_i\in\bV$, and now we have~$X_i\in\bV^d$. Perhaps use~$Z$? Consider reformulating: ``...such that for any given input~$Z=(Z_1,\ldots,Z_d)\in \bV^d$...''}
We first prove that %$f'$ satisfies the needed properties, i.e., 
$f'$ is multilinear with respect to the $d$ inputs.
%which is formally stated in the following lemma: 
%\blue{todo in short ver: mentioning the intuition, for the special case of real number field, this function can alternatively be derived }
%We start by proving the multilinearity of $f'$.
Recall that by definition, $f$ is a linear combination of monomials, and $f'$ is constructed based on $f$ through a linear operation. By exploiting the commutativity of these these two linear relations, we only need to show individually that each monomial in $f$ is transformed into a multilinear function.

More specifically, let $f$ be the sum of monomials $h_k\triangleq U_k  \cdot \prod\limits_{\ell=1}^{d_k} h_{k,\ell}(\cdot)$ where $k$ belongs to a finite set, $U_k\in\bU$, $d_k\in\{0,1,...,d\}$, and each $h_{k,\ell}$ is a linear map from $\bV$ to $\bF$. Let $h'_k$ denotes the contribution of $h_k$ in $f'$, then for any $Z=(Z_1,...,Z_d)\in \bV^d$ we have
\begin{align}
    h'_k(Z)&=\sum_{\mathcal{S}\subseteq[d]}{(-1)^{|\mathcal{S}|}h_k\left(\sum_{j\in\mathcal{S}}Z_{j}\right)}\nonumber\\
    &=\sum_{\mathcal{S}\subseteq[d]}{(-1)^{|\mathcal{S}|}U_k  \cdot \prod\limits_{\ell=1}^{d_k} h_{k,\ell}\left(\sum_{j\in\mathcal{S}}Z_{j}\right)}.
\end{align}

By utilizing the linearity of each $h_{k,\ell}$, we can write $h'_k$ as
\begin{align}
    h'_k(Z)&=U_k  \cdot \sum_{\mathcal{S}\subseteq[d]}{(-1)^{|\mathcal{S}|}\prod\limits_{\ell=1}^{d_k} \sum_{j\in\mathcal{S}}h_{k,\ell}\left(Z_{j}\right)}\nonumber\\
    &=U_k  \cdot \sum_{\mathcal{S}\subseteq[d]}{(-1)^{|\mathcal{S}|} \prod\limits_{\ell=1}^{d_k} \sum_{j=1}^d\mathbbm{1}{(j\in\mathcal{S})}\cdot h_{k,\ell}\left(Z_{j}\right)}
\end{align}
Then by viewing each subset $\mathcal{S}$ of $[d]$ as a map from $[d]$ to $\{0,1\}$, we have\footnote{Here we define $0^0=1$.}
% \begin{align}
%     h'_k(Z)&=U_k  \cdot\sum_{\boldsymbol{s}\in\{0,1\}^d}{\left(\prod_{m=1}^d(-1)^{{s_m}}\right) \prod\limits_{\ell=1}^{d_k} \sum_{j=1}^d s_j\cdot h_{k,\ell}\left(Z_{j}\right)}\nonumber\\
%     &=U_k  \sum_{\boldsymbol{j}\in[d]^{d_k}}\sum_{\boldsymbol{s}\in\{0,1\}^d}{\left(\prod_{m=1}^d(-1)^{{s_m}}\right) \cdot  \prod\limits_{\ell=1}^{d_k} (s_{j_\ell}\cdot h_{k,\ell}\left(Z_{j_\ell}\right))}.
% \end{align}
\begin{alignat}{2}
    h'_k(Z)&=U_k  \sum_{\boldsymbol{s}\in\{0,1\}^d}&&\left(\prod_{m=1}^d(-1)^{{s_m}}\right)\nonumber\\ &~&&\cdot\prod_{\ell=1}^{d_k} \sum_{j=1}^d s_j\cdot h_{k,\ell}\left(Z_{j}\right)\nonumber\\
    &=U_k  \sum_{\boldsymbol{j}\in[d]^{d_k}}&&\sum_{\boldsymbol{s}\in\{0,1\}^d}\left(\prod_{m=1}^d(-1)^{{s_m}}\right) \nonumber\\
    &~&&\phantom{\sum_{\boldsymbol{s}\in\{0,1\}^d}}\cdot  \prod_{\ell=1}^{d_k} (s_{j_\ell}\cdot h_{k,\ell}\left(Z_{j_\ell}\right)).
\end{alignat}
Note that the product $\prod\limits_{\ell=1}^{d_k} s_{j_\ell}$ can be alternatively written as $\prod\limits_{m=1}^d s_m^{\#(m~ \textup{in} ~\boldsymbol{j})}$, {where $\#(m~ \textup{in} ~\boldsymbol{j})$ denotes the number of elements in $\boldsymbol{j}$ that equals $m$. } Hence 
\begin{alignat}{3}\label{eq:conv_inter}
      % h'_k (Z)&=U_k  \cdot{ \sum_{\boldsymbol{j}\in[d]^{d_k}}  \sum_{\boldsymbol{s}\in\{0,1\}^d}\left(\prod_{m=1}^d\left((-1)^{{s_m}}  s_m^{\#(m~ \textup{in} ~\boldsymbol{j})}\right)\right)  \prod\limits_{\ell=1}^{d_k} h_{k,\ell}\left(Z_{j_\ell}\right)}\nonumber\\
     %&=U_k  \cdot{ \sum_{\boldsymbol{j}\in[d]^{d_k}}  \left(\prod_{m=1}^d\sum_{s\in\{0,1\}}(-1)^{{s}}  s^{\#(m~ \textup{in} ~\boldsymbol{j})}\right)  \prod\limits_{\ell=1}^{d_k} h_{k,\ell}\left(Z_{j_\ell}\right)}.\label{eq:conv_inter}
       h'_k (Z)&=U_k  \cdot \sum_{\boldsymbol{j}\in[d]^{d_k}}  &&\sum_{\boldsymbol{s}\in\{0,1\}^d}\left(\prod_{m=1}^d\left((-1)^{{s_m}}  s_m^{\#(m~ \textup{in} ~\boldsymbol{j})}\right)\right) \nonumber \\ &~&&\phantom{\sum_{\boldsymbol{s}\in\{0,1\}^d}}\cdot\prod_{\ell=1}^{d_k} h_{k,\ell}(Z_{j_\ell})\nonumber\\
      &=U_k \cdot \sum_{\boldsymbol{j}\in[d]^{d_k}}  &&\left(\prod_{m=1}^d \sum_{s\in\{0,1\}}(-1)^{s}  s^{\#(m~ \textup{in} ~\boldsymbol{j})}\right) \nonumber\\ &~&&\cdot\prod_{\ell=1}^{d_k} h_{k,\ell}(Z_{j_\ell}).
\end{alignat}
    The sum $\sum\limits_{s\in\{0,1\}}(-1)^{{s}}  s^{\#(m~ \textup{in} ~\boldsymbol{j})}$ is non-zero only if $m$ appears in $\boldsymbol{j}$. Consequently, among all terms that appear in (\ref{eq:conv_inter}), only the ones with degree $d_k=d$ and distinct elements in $\boldsymbol{j}$ have non-zero contribution. More specifically, \footnote{Here $S_d$ denotes the symmetric group of degree $d$.}
       % functions  $\prod\limits_{\ell=1}^{d_k} h_{k,\ell}\left(X_{i,j_\ell}\right)$ 
    \begin{align}
    h'_k(Z) &=\left(-1\right)^d \cdot \mathbbm{1}(d_k=d) \cdot U_k  \cdot{ \sum_{g\in S_d}    \prod\limits_{j=1}^{d} h_{k,g(j)}\left(Z_{j}\right)}.\label{eq:conv_finalh}
\end{align}
Recall that  $f'$ is a linear combination of $h'_k$'s. Consequently, it is a multilinear function. 

Now we prove that $f'$ is non-zero. From equation (\ref{eq:conv_finalh}), we can show that when all the elements $Z_{j}$'s are identical, $f'(Z)$ equals the evaluation of the highest degree terms of $f$ multiplied by a constant $(-1)^d d!$ with $Z_{j}$ as the input for any $j$. Given that the highest degree terms can not be zero, and $(-1)^d d!$ is non-zero as long as the characteristic of the field $\bF$ is greater than $d$, we proved that $f'$ is non-zero.

%\subsection{Proof of Lemma \ref{lemma:hamming}} \label{pl:hamming}
%\input{lemma_hamming.tex}

\subsection{Optimality in randomness}\label{app:pt_rand}

%Finally, the LCC scheme pads the dataset~$X$ with additional~$T$ random entries to guarantee $T$-privacy; and this amount of randomness is shown to be minimal.
In this appendix, we prove the optimality of LCC in terms of the amount of randomness needed in data encoding, which is formally stated in the following theorem.
 \begin{theorem} \label{lemma:optimalRandomness}
 (Optimal randomness) Any linear encoding scheme that universally achieves a same tradeoff point %\footnote{{I.e., \Oldver{when the conditions in Definition} \ref{def:r} holds equal.}} 
  specified in Theorem \ref{thm:lcc}  for all linear functions~$f$ (i.e., $(S,A,T)$ such that $K+T+S+2A= N$) %(and in particular to the identity function)
 must use an amount of randomness no less than that of LCC.
%(Optimal randomness) Any linear coded computing scheme that achieves the tradeoff of Theorem~\ref{theorem:optimalRateNoPrivacy} (i.e., when two sides of (\ref{ineq:conv}) are equal) must use an amount of randomness no less than that of LCC.
\end{theorem}

%\Qian{TODO: move}
%The proof of Theorem~\ref{lemma:optimalRandomness} is given {in Appendix \ref{app:pt_rand}}. %\textcolor{blue}{[Change if the proof of Theorem 4 was deferred to the appendix. --NR]}.

%In this appendix we prove the optimality of private Lagrange Coded Computing in terms of the amount of randomness that is used. That is, it is proved that no linear encoding scheme which universally achieves the optimal computing threshold for all linear computations can achieve robustness against colluding sets of workers of size~$T$ by padding with less than~$T$ random entries. 

%In this appendix we prove that LCC uses the minimum possible randomness among all linear schemes that achieves the optimum tradeoff stated in Theorem \ref{thm:lcc} for linear $f$.  

%the optimality of private Lagrange Coded Computing in terms of the amount of randomness that is used. That is, it is proved that no linear encoding scheme which universally achieves the optimal computing threshold for all linear computations can achieve robustness against colluding sets of workers of size~$T$ by padding with less than~$T$ random entries. 
\begin{proof}

The proof is taken almost verbatim from~\cite{WentaosThesis}, Chapter~3. In what follows, an~$(n,k,r,z)_{\bF_q^t}$ \textit{secure RAID scheme} is a storage scheme over~$\bF_q^t$ (where~$\bF_q$ is a field with~$q$ elements) in which~$k$ message symbols are coded into~$n$ storage servers, such that the~$k$ message symbols are reconstructible from any~$n-r$ servers, and any~$z$ servers are information theoretically oblivious to the message symbols. Further, such a scheme is assumed to use~$v$ random entries as keys, and by~\cite{WentaosThesis}, Proposition~3.1.1, must satisfy~$n-r\ge k+z$.

\begin{theorem}\cite{WentaosThesis}, Theorem~3.2.1. \label{theorem:Wentao}
A linear rate-optimal~$(n,k,r,z)_{\bF_q^t}$ secure RAID scheme uses at least~$zt$ keys over~$\bF_q$ (i.e.,~$v\ge z$).
\end{theorem}

Clearly, in our scenario~$\bV$ can be seen as~$\bF_q^{\dim\bV}$ for some~$q$. Further, by setting~$N=n$, $T=z$, and~$t=\dim\bV$, it follows from Theorem~\ref{theorem:Wentao}
that any encoding scheme which guarantees information theoretic privacy against sets of~$T$ colluding workers must use at least~$T$ random entries~$\{Z_i\}_{i\in[T]}$.
\end{proof}

%\subsection{Proof of Lemma \ref{lemma:t0}} \label{pl:t0}
%\input{lemma_t0.tex}
%\Removed{\subsection{Proof of Lemma \ref{lemma:t1}} \label{pl:t1}
%\input{Lemma_t1.tex}}

\subsection{Optimality of LCC for Linear Regression }\label{sec:regression_converse}
% In this section we prove the converse bound for Corollary~1 of the paper. 
In this section, we prove that the proposed LCC scheme achieves the minimum possible recovery threshold $R^*$ to within a factor of 2, for the linear regression problem discussed in Section~6.

As the first step, we prove a lower bound on $R^*$ for linear regression. More specifically, we show that for any coded computation scheme, the master always needs to wait for at least $\lceil\frac{n}{r}\rceil$ workers to be able to decode the final result, i.e., $R^* \geq \lceil\frac{n}{r}\rceil$. Before starting the proof, we first note that since here we consider a more general scenario where workers can compute \emph{any} function on locally stored coded sub-matrices (not necessarily matrix-matrix multiplication), the converse result in Theorem~\ref{thm:opt} no longer holds. 

% The factor of two characterization in Corollary~1 of the paper directly follows from this lower bound, as $\tfrac{1}{2}K_{\textup{LCC}}<\lceil\frac{n}{r}\rceil$.

To prove the lower bound, it is equivalent to show that, for any coded computation scheme and any subset $\mathcal{N}$ of workers, if the master can recover $\vct{X}^\top\vct{X}\vct{w}$ given the results from workers in $\mathcal{N}$, then we must have $|\mathcal{N}|\geq \lceil\frac{n}{r}\rceil$.
Suppose the condition in the above statement holds, then we can find encoding, computation, and decoding functions such that for any possible values of $\vct{X}$ and $\vct{w}$, the composition of these functions returns the correct output.

% Note that within a GD iteration, the workers compute some intermediate results only based on a vector $\vct{w}$ and their locally stored coded sub-matrices. Hence, if the master can decode the final output from the results of the workers in a subset $\mathcal{N}$, then the composition of the decoding function and the computing functions of these workers essentially computes $\vct{X}^\top\vct{X}\vct{w}$ only using only the coded sub-matrices stored by these workers and the vector $\vct{w}$. Hence, if any class of input values $\vct{X}$ gives the same coded matrices for any worker in $\mathcal{N}$, then the product $\vct{X}^\top\vct{X}\vct{w}$ must also be the same given any $\vct{w}$. 

Note that within a GD iteration, %the workers compute some intermediate results only based on a vector $\vct{w}$ and their locally stored coded sub-matrices.
each worker performs its local computation only based on its locally stored coded sub-matrices and the weight vector $\vct{w}$. Hence, if the master can decode the final output from the results of the workers in a subset $\mathcal{N}$, then the composition of the decoding function and the computation functions of these workers essentially computes $\vct{X}^\top\vct{X}\vct{w}$, using only the coded sub-matrices stored at these workers and the vector $\vct{w}$. Hence, if any class of input values $\vct{X}$ gives the same coded sub-matrices for each worker in $\mathcal{N}$, then the product $\vct{X}^\top\vct{X}\vct{w}$ must also be the same given any $\vct{w}$. 

%these  the output $\vct{X}^\top\vct{X}\vct{w}$ must be computable based on the coded sub-matrices stored by these workers and the vector $\vct{w}$. %$\tilde{\vct{X}}_{j,k}$, for any $j\in \mathcal{N}$ and $k\in\{0,1,\ldots,r-1\}$ 
%TODO:revise

Now we consider the class of input matrices $\vct{X}$ such that all coded sub-matrices stored at workers in $\mathcal{N}$ equal the values of the corresponding coded sub-matrices when $\vct{X}$ is zero. %, and we denote this set of values by $\mathcal{X}$. 
Since $\vct{0}^\top\vct{0}\vct{w}$ is zero for any $\vct{w}$, 
%Because the coded matrices are fixed, the function $\vct{X}^\top\vct{X}\vct{w}$ can only depends on $\vct{w}$ for any such value of  $\vct{X}$. Moreover, because the zero matrix belong to this class, 
$\vct{X}^\top\vct{X}\vct{w}$ must also be zero for all matrices $\vct{X}$ in this class and any $\vct{w}$. However, for real matrices $\vct{X}=\vct{0}$ is the only solution to that condition. Thus, zero matrix must be the only input matrix that belongs to this class.

%kernel to the collection of encoding functions for workers in $\mathcal{N}$ must be zero.  

Recall that all the encoding functions are assumed to be linear. We consider the collection of all encoding functions that are used by workers in $\mathcal{N}$, which is also a linear map. As we have just proved, the kernel of this linear map is $\{\vct{0}\}$. Hence, its rank must be at least the dimension of the input matrix, which is $dm$.
%Thus the total rank of the encoding functions for workers in $\mathcal{N}$ must be the same as the number of entries in $\vct{X}$, which is $dm$. 
On the other hand, its rank is upper bounded by the dimension of the output, where each encoding function from a worker contributes at most 
%the encoding functions of each worker has a total rank of at most 
$\frac{rdm}{n}$. Consequently, the number of workers in $\mathcal{N}$ must be at least $\lceil\frac{n}{r}\rceil$ to provide sufficient rank to support the computation. %This concludes the converse proof.

Having proved that $R^* \geq \lceil\frac{n}{r}\rceil$, the factor of two characterization of LCC directly follows since  $R^* \leq R_{\textup{LCC}} = 2\lceil \tfrac{n}{r}\rceil -1<2\lceil\frac{n}{r}\rceil \leq 2R^*$.

Note that the converse bound proved above applies to the most general computation model, i.e., there are no assumptions made on the encoding functions or the functions that each worker computes. If additional requirements are taken into account, we can show that LCC achieves the exact optimum recovery threshold (e.g., see \cite{yu2018straggler}).   

\subsection{Complete Experimental Results}\label{sec:experiments}
In this section, we present the complete experimental results using the LCC scheme proposed in the paper, the gradient coding (GC) scheme~\cite{TandonLDK17} (the cyclic repetition scheme), the matrix-vector multiplication based (MVM) scheme~\cite{lee2015speeding}, and the uncoded scheme for which there is no data redundancy across workers, measured from running linear regression on Amazon EC2 clusters.

In particular, experiments are performed for the following 3 scenarios.
\begin{itemize}[leftmargin=*]
    \item Scenario 1 \& 2: \# of input data point $m = 8000$, \# of features $d = 7000$. 
    %\vspace{-2mm}
    \item Scenario 3: \# of input data point $m = 160000$, \# of features $d = 500$.
\end{itemize}
%\vspace{-3mm}

In scenarios 2 and 3, we artificially introduce stragglers by imposing a $0.5$ seconds delay on each worker with probability $5\%$ in each iteration.

%listed in Table~\ref{table:scenarios_supp}. 

% \begin{table}[htbp]
% \vspace{-5mm}
% \label{table:scenarios_supp}
% \caption{Experiment scenarios.}
%   \centering
%   \scalebox{0.68}{
%   \begin{tabular}{| c | c | c | c | }
%     \hline
%     scenario index & \# of inputs ($m$)& \# of features ($d$) &artificial stragglers  \\ \hline
%     1 & 8000  &7000  & No\\ \hline
%     2 & 8000  &7000  & Yes\\ \hline
%     3 & 160000  &500  & No\\ \hline
%     4 & 160000  &500  & Yes\\ \hline
%   \end{tabular}}
% \end{table}

We list the detailed breakdowns of the run-times in 3 experiment scenarios in Tables~\ref{table:scenario one_supp}, \ref{table:scenario two_supp}, and \ref{table:scenario three_supp} respectively. In particular, the computation (comp.) time is measured as the summation of the maximum local processing time among all non-straggling workers, over 100 iterations. The communication (comm.) time is computed as the difference between the total run-time and the computation time.

% We list the breakdowns of the run-times of the uncoded, GC, MVM, and LCC schemes in four experiment scenarios in Tables~\ref{table:scenario one_supp}, \ref{table:scenario two_supp}, \ref{table:scenario three_supp}, and \ref{table:scenario four_supp} respectively. 
%We also plot the empirical CDFs of the per iteration run-time for the PCR and BCC schemes in all four experiment scenarios in Figure~\ref{fig:cdf}.

\begin{table}[!htbp]
% \vspace{-5mm}
 \caption{Breakdowns of the run-times in scenario one.}
\label{table:scenario one_supp}
  \centering
  \scalebox{0.72}{
  \begin{tabular}{|c | c | c| c | c | c |}
    \hline
     \multirow{2}{*}{schemes} & \# batches/ & recovery & comm.  &comp.  & total \\
     & worker ($r$) & threshold & time & time& run-time \\
     \hline
    uncoded & 1 & 40   &24.125 s &0.237 s &  24.362 s\\ \hline
    GC & 10 &31&6.033 s  & 2.431 s  &8.464 s\\ \hline
    MVM Rd. 1 & 5 & 8&1.245 s & 0.561 s&  1.806 s\\ \hline
    MVM Rd. 2 & 5 & 8&1.340 s & 0.480 s&  1.820 s\\ \hline
    MVM total & 10 & -&2.585 s & 1.041 s&  3.626 s\\ \hline
    LCC & 10 & 7&1.719 s & 1.868 s&  3.587 s\\ \hline
  \end{tabular}}
%  \vspace{-3mm}
  %\newline\newline
\end{table}

% \begin{table}[!htbp]
%  \caption{Breakdowns of the run-times in scenario two with $n=40$ workers.}
%   \vspace{0.5mm}
% \label{table:scenario two}
%   \centering
%   \scalebox{0.63}{
%   \begin{tabular}{|c | c | c| c | c | c |}
%     \hline
%      \multirow{2}{*}{schemes} & \# batches processed & recovery & communication  &computation  & total \\
%      & at each worker ($r$) & threshold & time & time& run-time\\
%      \hline
%     uncoded & 1 & 40 & 7.928 s & 44.772 s & 52.700 s\\ \hline
%     GC & 10 & 31 & 14.42 s  & 2.401 s  & 16.821 s\\ \hline
%     PCR & 10 & 7 & 2.019 s & 1.906 s &  3.925 s\\ \hline
%   \end{tabular}}
%   %\newline\newline
% \end{table}
\begin{table}[!htbp]
% \vspace{-3mm}
 \caption{Breakdowns of the run-times in scenario two.}
\label{table:scenario two_supp}
  \centering
  \scalebox{0.72}{
  \begin{tabular}{|c | c | c| c | c | c |}
    \hline
     \multirow{2}{*}{schemes} & \# batches/ & recovery & comm.  &comp.  & total \\
     & worker ($r$) & threshold & time & time& run-time\\
     \hline
    uncoded & 1 & 40 & 7.928 s & 44.772 s & 52.700 s\\ \hline
    GC & 10 & 31 & 14.42 s  & 2.401 s  & 16.821 s\\ \hline
    MVM Rd. 1 & 5 & 8 & 2.254 s & 0.475 s &  2.729 s\\ \hline
    MVM Rd. 2 & 5 & 8 & 2.292 s & 0.586 s &  2.878 s\\ \hline
    MVM total & 10 & - & 4.546 s & 1.061 s &  5.607 s\\ \hline
    LCC & 10 & 7 & 2.019 s & 1.906 s &  3.925 s\\ \hline
  \end{tabular}}
%  \vspace{-3mm}
  %\newline\newline
\end{table}

\begin{table}[!htbp]
% \vspace{-3mm}
\caption{Breakdowns of the run-times in scenario three.}
\label{table:scenario three_supp}
  \centering
\scalebox{0.71}{
  \begin{tabular}{|c | c | c| c | c | c |}
    \hline
     \multirow{2}{*}{schemes} & \# batches/ & recovery & comm.  &comp.  & total \\
     &worker ($r$) & threshold & time & time& run-time\\
     \hline
    uncoded & 1 & 40 & 0.229 s & 41.765 s & 41.994 s\\ \hline
    GC & 10 & 31 & 8.627 s & 2.962 s  & 11.589 s\\ \hline
    MVM Rd. 1 & 5 & 8 & 3.807 s & 0.664 s&  4.471 s\\ \hline
    MVM Rd. 2 & 5 & 8 & 52.232 s & 0.754 s&  52.986 s\\ \hline
    MVM total & 10 & - & 56.039 s & 1.418 s&  57.457 s\\ \hline
    LCC & 10 & 7 & 1.962 s & 2.597 s&  4.541 s\\ \hline
  \end{tabular}}
 % \vspace{-5mm}
  %\newline\newline
\end{table}

%\subsection{Discussion on optimality of LCC}
 %\label{app:discussion}
%\input{Conclusion.tex}
\end{document}